\documentclass{llncs}
\usepackage{times,amsmath,epsfig,amssymb,color}
\usepackage[named]{algo}

\newcommand{\ie}{{\it i.e.}}

\def\papernumber #1 raised #2 {
\vspace{-#2}
\vbox to 0pt{\hfill\framebox{\bf Paper Number #1}}
\vspace{#2}
}

\newsavebox{\savepar}

\newcommand{\squishlist}{
  \begin{list}{$\bullet$}
   {
     \setlength{\itemsep}{0mm}
     \setlength{\topsep}{3pt}
     \setlength{\leftmargin}{1.0em}
     \setlength{\labelwidth}{1em}
     \setlength{\labelsep}{0.5em} } }

\newcommand{\squishend}{
   \end{list}  }

\newcommand{\deltu}
{\mbox{$\mathrel{\hbox{$\bigtriangleup$\raise1.5pt\hbox{\hskip-6.5pt
{\tiny $\mu$}}}}$\,}}
\newcommand{\delu}
{\mbox{$\mathrel{\hbox{$\bigtriangledown$\raise2pt\hbox{\hskip-6.5pt
{\tiny $\mu$}}}}$\,}}

\newcommand{\monus}{\mbox{$\mathrel{\hbox{$-$\raise2pt\hbox{\hskip
-5.5pt$\cdot$}}}$\,}}

\newcommand{\semijoin}{\mbox{$\mathrel{\raise1pt\hbox{\vrule height5pt
depth0pt\hskip-1.5pt$>$\hskip -2.5pt$<$}}$}}

\usepackage{euscript}

\newbox\ProofSym
\setbox\ProofSym=\hbox{%
\unitlength=0.18ex%
\begin{picture}(10,10)
\put(0,0){\framebox(9,9){}}
\put(0,3){\framebox(6,6){}}
\end{picture}}

\begin{document}
%
%
\title{Spatial Skyline Queries:\\An Efficient Geometric Algorithm}
\titlerunning{Spatial Skyline Queries}  
%
\author{Wanbin Son \and Mu-Woong Lee \and Hee-Kap Ahn \and Seung-won Hwang}
\authorrunning{Wan-Bin Son et al.}   
%
\tocauthor{Wan-Bin Son, Mu-Woong Lee, Hee-Kap Ahn, Seung-won Hwang}
\institute{Pohang University of Science and Technology, Korea\\
\email{\{mnbiny, sigliel, heekap, swhwang@postech.ac.kr\}}}

\maketitle              

\newtheorem{assumption}{Assumption}

\newcommand{\comment}[1]{\textcolor{red}{\uppercase{\# #1 \#}}}
\renewcommand{\comment}[1]{}

\newcommand{\dist}[1]{\ensuremath{d(#1)}}
\newcommand{\CH}[1]{\ensuremath{{\cal CH}(#1)}}
\newcommand{\A}{\ensuremath{{\cal A}}}
\newcommand{\VD}[1]{\ensuremath{\mathrm{Vor}(#1)}}
\newcommand{\VC}[1]{\ensuremath{{\cal V}(#1)}}
\newcommand{\rtree}{{R$^{\ast}$-tree}}
\newcommand{\vstwo}{\emph{VS$^{2}$}}
\newcommand{\esky}{\emph{ES}}
\newcommand{\asky}{\emph{AS}}

\newcommand{\argmax}{\operatornamewithlimits{argmax}} 
\newcommand{\argmin}{\operatornamewithlimits{argmin}} 

\begin{abstract}

  As more data-intensive applications emerge, advanced retrieval
  semantics, such as ranking or skylines, have attracted attention.
  Geographic information systems are such an application with massive
  spatial data.  Our goal is to efficiently support 
  skyline queries over massive spatial data. To achieve
  this goal, we first observe that 
  the best known algorithm \vstwo, despite its claim, 
  may fail to deliver correct results.  In
  contrast, we present a simple and efficient algorithm
  that computes the correct results.
  To validate the effectiveness and efficiency
  of our algorithm, we provide an extensive empirical comparison of
  our algorithm and \vstwo\ in several aspects.

\end{abstract}

\section{Introduction}\label{sec:intro}

With the advent of data-intensive applications, advanced query semantics,
which enable efficient and intelligent access to a large scale data, have been
actively studied lately.  Geographic information systems (GIS) are such an
application, which aims at supporting efficient access to massive spatial data,
as Example \ref{ex:dr} illustrates.

\begin{figure}[b]
	\centering
	\begin{tabular}{c c c}
		\epsfig{file=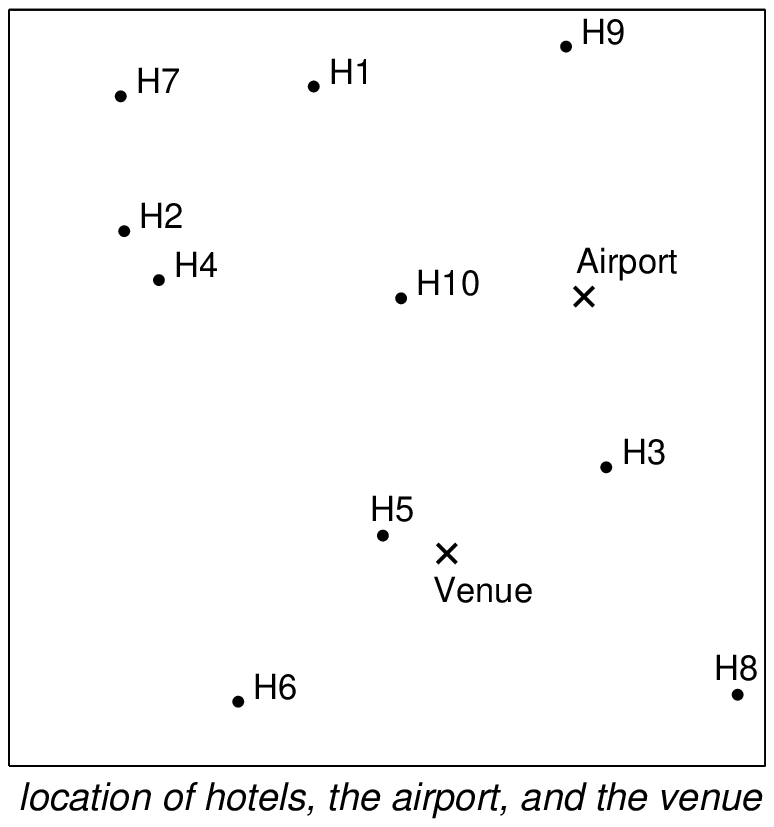, height=48mm, clip=} & \ \ &
		\epsfig{file=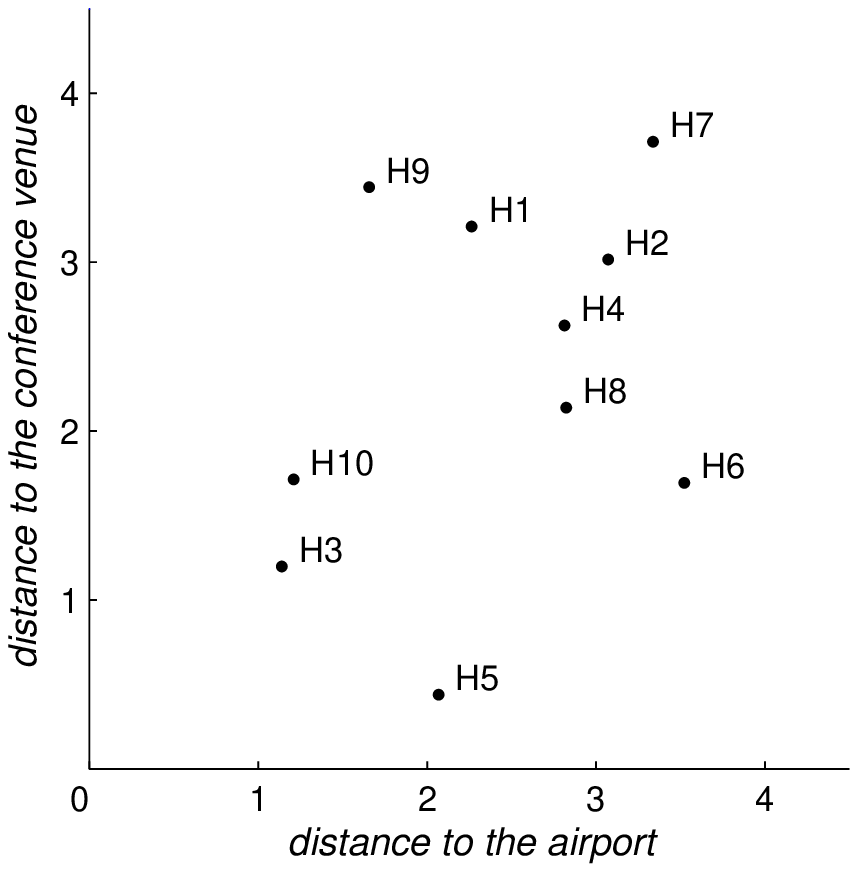, height=48mm, clip=}
	\end{tabular}
	\caption{Hotel search scenario}
	\label{fig:ex}
\end{figure}

\begin{example}\label{ex:dr}
Consider a hotel search scenario for a business trip to San Francisco, where
the user marks two locations of interest, e.g., the conference venue and an 
airport, as Fig.~\ref{fig:ex}(left) illustrates.  Given these two query 
locations, it would
be interesting to identify hotels that are close to both locations.  To better
illustrate this problem, Fig.~\ref{fig:ex}(right) 
rearranges the hotels
with respect to the distance to each query point.  From this figure, we can
claim that hotel H3 is more desirable than H10, because H3 is closer to 
both query points than H10 is. Such advanced retrieval, by \emph{ranking} 
the hotels using the aggregate distance to the given query points, or 
by finding \emph{skyline} hotels, 
will enable intelligent access to the underlying hotel datasets.
\end{example}

In particular, this paper focuses on  supporting \emph{skyline queries}
\cite{Kung75,borzsonyi01,Tan01,Papadias03,Chomicki03} to identify the objects
that are ``not dominated'' by any other objects, \ie, no other object is 
closer to all the given query points simultaneously. 
For instance, in Fig.~\ref{fig:ex}(right), H3 is a
skyline object, while H10 is dominated by H3 
and does not qualify as a skyline object. 

Skyline queries have gained attention lately, as formulating such queries
is highly intuitive, compared to ranking where users are required to identify
ideal distance functions to minimize.
However, most of existing skyline algorithms have not been devised for spatial
data and thus do not consider spatial relationships between objects.

Our goal is to efficiently support skyline queries over spatial data. 
This problem has already been studied by Sharifzadeh and 
Shahabi~\cite{Sharifzadeh} and they presented two algorithms for the 
problem, one of which, \vstwo, is known to be the most efficient solution
thus far. We claim, however, that \vstwo\ may fail to 
identify the correct results. In a clear contrast, we propose an algorithm 
for the problem that can identify the exact results in
$O(|P|(|S|\log|\ensuremath{{\cal CH}(Q)}|+ \log|P|))$ time, 
for the given set $P$ of data points, set $Q$ of query points, set $S$ of
spatial skylines, and the \emph{convex hull} of $Q$, denoted by
\ensuremath{{\cal CH}(Q)}.


Our contributions can be summarized as follows:
\begin{itemize}
\item We study the spatial skyline query processing problem, which enables
intelligent and efficient access to massive spatial data.

\item We show that the best known algorithm is incomplete in the sense
  that it may not return all the skyline points. 

\item We propose a novel and correct spatial skyline query processing 
  algorithm and analyze its complexity.

\item We extensively evaluate our framework using synthetic data and
  validate its effectiveness. 
\end{itemize}

The remainder of this paper is organized as follows. In
Section~\ref{sec:related}, we provide a brief survey on related work.
In Section~\ref{sec:prel}, we observe the drawbacks in the best known
algorithm as preliminaries and propose a new algorithm in 
Section~\ref{sec:algo}.
Section \ref{sec:imp} discusses the details of our implementation of the
proposed algorithm.  In Section~\ref{sec:exp}, we report our evaluation
results. 

\section{Related Work}\label{sec:related}

This section provides a brief survey on work related to (1) skyline
query processing and (2) spatial query processing.

{\bf Skyline computation:} Skyline queries were first
studied as maximal vectors in~\cite{Kung75}. Later,
B{\"o}rzs{\"o}nyi at el.~\cite{borzsonyi01} introduced skyline
queries in database applications. A number of
different algorithms for skyline computation have been proposed. 
For example, Tan et al. \cite{Tan01} (progressive skyline computation 
using auxiliary structures), Kossmann et al.~\cite{Kossmann02} 
(nearest neighbor algorithm for skyline query processing),
Papadias et al.~\cite{Papadias03} (branch and bound skyline (BBS) algorithm),
Chomicki et al.~\cite{Chomicki03} (sort-filter-skyline (SFS) algorithm
leveraging pre-sorting lists), and Godfrey et al. \cite{Godfrey05}
(linear elimination-sort for skyline (LESS) algorithm with
attractive average-case asymptotic complexity). Recently, there have
been active research efforts to address the ``curse of dimensionality"
problem of skyline queries \cite{ChanED06,Chan06,Lin07} using
inherent properties of skylines such as \emph{skyline frequency},
\emph{k-dominant skylines}, and \emph{k-representative skylines}.
All these efforts, however, do not consider
spatial relationships between data objects.

{\bf Spatial query processing:}
The most extensively studied spatial query mechanism is ranking
the neighboring objects by the distance to the single query point
\cite{Rousso,Berchtold,Bohm}.
For multiple query points,
Papadias et al. \cite{Papadias} studied ranking by the ``aggregate"
distance, for a class of monotone functions
aggregating the distances to multiple query points.
As these nearest neighbor queries require distance function,
which is often cumbersome to define,
another line of research studied skyline query semantics which
do not require such functions.
For a spatial skyline query with a single query point, 
Huang and Jensen~\cite{Huang}
studied the problem of finding spatial locations that are not
dominated with respect to the \emph{network distance} to the query point.
For such query with multiple query points, 
Sharifzadeh and Shahabi~\cite{Sharifzadeh}
proposed two algorithms that identify the skyline locations to the 
given query points such that no other location is closer to all query points.
While the proposed problem enables intelligent access to
spatial data, we later show that the solution proposed 
in~\cite{Sharifzadeh} is incorrect.
In contrast, this paper presents a correct exact algorithm.

\section{Preliminaries}\label{sec:prel}
In this section, we introduce some geometric concepts
(Section~\ref{subsec:convex} and \ref{subsec:voro}),
and define our problem (Section~\ref{subsec:pd}).
Then we discuss how the best known algorithm fails to identify the exact answers
(Section~\ref{subsec:naive}).

\subsection{Convex Hull}\label{subsec:convex}
\comment{Describe Convex hull briefly. Lower chain? Upper chain?}

A subset $S$ of the plane is \emph{convex} if and only if for every
two points $p,q \in S$ the whole line segment $\overline{pq}$ is
contained in $S$.  The \emph{convex hull} \CH{S} of a set $S$ is the
intersection of all convex sets that contains $S$~\cite{CG}.  The
\emph{upper chain} of \CH{S} is the part of the boundary of \CH{S} 
from the leftmost point to the rightmost point in clockwise order. 
The \emph{lower chain} is the part of the boundary of \CH{S} 
from the rightmost point to the leftmost point in counterclockwise order. 

\subsection{Voronoi Diagram and Delaunay Graph}\label{subsec:voro}
\comment{Describe voronoi diagram briefly.}

For a set $P$ of $n$ distinct points in the plane, the Voronoi diagram of $P$,
denoted by \VD{P}, is the subdivision of the plane into $n$
cells~\cite{CG} . Each cell contains only one point of $P$, which is 
called the \emph{site} of the cell. Any point $q$ in a cell
is closer to the site of the cell than any other site.
The Delaunay graph of a point set $P$ is the dual graph of the Voronoi
diagram of $P$~\cite{CG}. Two points of $P$ have an edge in the Delaunay
graph if and only if the Voronoi cells of these points share an edge
in \VD{P}.

\subsection{Problem Definition}\label{subsec:pd}
In the spatial skyline query problem, we are given two point sets: one is a set
$P$ of data points, and the other is a set $Q$ of query points.  The points in
$P$ and $Q$ have $d$-dimensional coordinate attributes in $\mathbb{R}^{d}$
space.  The distance function $\dist{p,q}$ returns the Euclidean distance
between a pair of points $p$ and $q$, which obeys the triangle inequality.
Before we set the goal of the problem, we need the following definitions.  

\begin{definition}\label{def:dominate}
  We say that $p_1$ \emph{spatially dominates} $p_2$ if and only if
  $\dist{p_1,q} \leq \dist{p_2,q}$ for every $q\in Q$, and
  $\dist{p_1,q'} < \dist{p_2,q'}$ for some $q'\in Q$.
\end{definition}

\begin{definition}\label{def:skyline}
  A point $p\in P$ is a \emph{spatial skyline point} with respect to
  $Q$ if and only if $p$ is not spatially dominated by any other point
  of $P$.
\end{definition}

\noindent The goal of the problem is to retrieve all the spatial skyline points
from $P$ with respect to $Q$. We denote by $S$ the set of spatial
skyline points of $P$

\subsection{Existing Approaches}\label{subsec:naive}
Though there is a lot of work on skyline queries in literature,
little has been known on the skyline queries for spatial data.
Recently, Sharifzadeh and Shahabi~\cite{Sharifzadeh} studied the
spatial skyline query problem and proposed two algorithms that compute $S$:
Branch-and-Bound Spatial Skyline Algorithm (B$^{2}$S$^{2}$) and
Voronoi-based Spatial Skyline Algorithm (\vstwo).

In \vstwo, they employed two well-known geometric structures, the
\emph{Voronoi diagram} of $P$ and the
\emph{convex hull} of $Q$, and claimed that these
structures reflect the spatial dominance to some extent, and therefore
the algorithm efficiently computes $S$.  In fact, their experiments
show that \vstwo\ runs $2\sim 3$ times faster than B$^{2}$S$^{2}$,
and \vstwo\ is known to be the most efficient solution thus far.

\comment{brief algorithm of \vstwo. which part of the algorithm is incomplete.}

\vstwo, however, may fail to find all the spatial skyline points:
In Lemma 4 of ~\cite{Sharifzadeh}, to verify \vstwo~they claimed that, 
for some $p\in P$, if all
its Voronoi neighbors and all their Voronoi neighbors are spatially
dominated by other points, $p$ is not a spatial skyline. Therefore
\vstwo\ simply marks $p$ as \emph{dominated} and does not consider
it afterwards. But this is not necessarily true.

\begin{figure}[ht]
  \begin{center}
    \epsfig{file=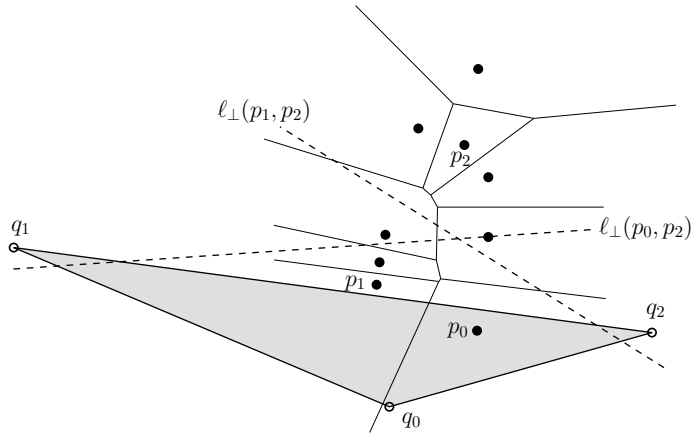, width=0.75\linewidth}
    \caption{VS$^2$ fails to find $p_2$ even though $p_2$ is a spatial
      skyline point}
    \label{fig:vs2_dominance}
  \end{center}
\end{figure}

Fig.~\ref{fig:vs2_dominance} shows a counter example to their claim.
There are $3$ query points ($q_0, q_1, q_2$) and $9$ data points.
Note that all the data points, except three ($p_0, p_1$ and $p_2$),
are spatially dominated by $p_0$ or $p_1$.  That is, all the Voronoi
neighbors of $p_2$ are spatially dominated, and \vstwo\ thus simply
marks $p_2$ as ``dominated'' and does not consider it again.
However, in fact, $p_2$ is a spatial skyline point, as the
\emph{bisector}  $\ell_\perp(p_1,p_2)$ of $p_1$ and $p_2$, \ie, a
perpendicular line to the line segment $\overline{pq}$, intersects
$\CH{Q}$.  This implies that there is a query point ($q_2$) closer to
$p_2$ and therefore $p_2$ is not spatially dominated by $p_1$, as we
will discuss more formally later in Lemma
\ref{lemma:geometric-property}.  Similarly, $p_2$ is not spatially
dominated by $p_0$, because $\ell_\perp(p_0,p_2)$ intersects
$\CH{Q}$. Since every bisecting line of $p_2$ and other points
intersects $\CH{Q}$, we conclude that $p_2$ is a spatial skyline
point.

\comment{a large scale example of VS2 failures.}

\begin{figure}[ht]
  \begin{center}
    \epsfig{file=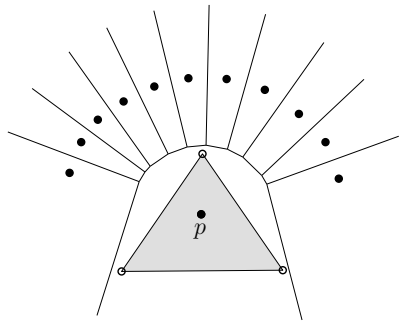, width=.43\linewidth}
    \caption{A point can have many neighbors}\label{fig:vs2_complexity}
  \end{center}
\end{figure}

Moreover, the asymptotic time 
complexity analysis of \vstwo\ in \cite{Sharifzadeh} is incorrect. 
The authors assumed implicitly that 
\vstwo\ tests only $O(|S|)$ points and claimed that it finds $S$ in time
$O(|S|^2|\CH{Q}|+\sqrt{|P|})$. However, a skyline point $p$ can
have at most $O(|P|)$ Voronoi neighbors that are all spatially
dominated by $p$, as Fig.~\ref{fig:vs2_complexity} illustrates. 
Since it also calls $|P|$ heap operations during the iteration,
each of which takes $\log|P|$, 
the correct worst-case time complexity of \vstwo\ must be
$O(|P|(|S||\CH{Q}|+\log|P|))$.

\section{Computing Spatial Skylines}\label{sec:algo}
We first propose a progressive algorithm for the spatial skyline problem, 
which retrieves all the spatial skyline points of $P$ with respect to $Q$,
then we improve this algorithm by using the Voronoi diagram of the dataset.

We assume the dimensionality $d$ of data and query points as $d=2$ 
for now, which can be extended for arbitrary dimension 
(as we will discuss in Section \ref{sec:con}). 

Before we explain our algorithms, we show some properties of spatial
skyline that will be used later on.
The following lemma is the contraposition of Definition~\ref{def:dominate}.

\begin{lemma}\label{lemma1}
  $p_1$ does not spatially dominate $p_2$ if and only if
  either $\dist{p_1,q}>\dist{p_2,q}$ for some $q\in Q$,
  or $\dist{p_1,q}=\dist{p_2,q}$ for every $q\in Q$.
\end{lemma}

\begin{lemma}\label{lemma3}
  Let $p_1, p_2$ and $p_3$ be three data points such that $p_2$ spatially
  dominates $p_3$. If $p_1$ does not spatially dominate $p_3$, it does
  not spatially dominate $p_2\in P$.
\end{lemma}

\begin{proof}
  Since $p_1$ does not spatially dominate $p_3$,
  either (1) $\dist{p_3,q'}<\dist{p_1,q'}$ for some $q'\in Q$,
  or (2) $\dist{p_3,q}\leq \dist{p_1,q}$ for every $q\in Q$
  by Lemma~\ref{lemma1}.
  \newline{Case (1)}. By Def~\ref{def:dominate},
  $\dist{p_2,q}\leq \dist{p_3,q}$ for every $q\in Q$.  This implies
  that $\dist{p_2,q'}\leq \dist{p_3,q'}<\dist{p_1,q'}$. Therefore,
  $p_1$ does not spatially dominate $p_2$ by Lemma~\ref{lemma1}.
  \newline{Case (2)}. Since $p_2$ spatially dominates $p_3$, there
  exists a point $q\in Q$ satisfying $\dist{p_2,q}<\dist{p_3,q}$,
  which implies that $\dist{p_3,q}\leq \dist{p_1,q}$.
  Therefore, $p_1$ does not spatially dominate $p_2$ by Lemma~\ref{lemma1}.
\end{proof}

\begin{lemma}\label{lemma4}
  If some data point $p_1$ is not a spatial skyline point, there
  always exists a spatial skyline point $p_2$ that spatially dominates
  $p_1$.
\end{lemma}

\begin{proof}  Since $p_1$ is not a spatial
  skyline point, there exists some data point that spatially dominates
  $p_1$.  Let $P'$ be the set of the data points that spatially
  dominate $p_1$, and let $p_2$ be the point which has the minimum sum of
  distances to all $q\in Q$ among points in $P'$. Then it is not
  difficult to see that for every point $p'\in P'$, there always
  exists some query point $q$ such that
  $\dist{p_2,q}<\dist{p',q}$. Therefore, $p_2$ is not spatially
  dominated by any point in $P'$. By Lemma~\ref{lemma3}, $p_2$ is not
  spatially dominated by any data point which does not spatially
  dominate $p_1$. This means that $p_2$ is not spatially dominated by any
  other data points, so $p_2$ is a spatial skyline point.
\end{proof}

We now move on to discuss how to use these properties to
reduce (1) the time required for each dominance test and (2) the
number of dominance tests.

\subsection{Efficient Spatial Dominance Test}\label{subsubsec:efficient_check}
Sharifzadeh and Shahabi~\cite{Sharifzadeh} showed that we can
determine spatial dominance by using just the convex hull of $Q$
instead of all query points in $Q$: If $p\in P$ is not dominated
by any other point in $P$ with respect to the vertices of \CH{Q},
then $p$ is a spatial skyline point.
In fact, we can interpret this property in a geometric setting as follows.

\begin{figure}[ht]
  \begin{center}
    \epsfig{file=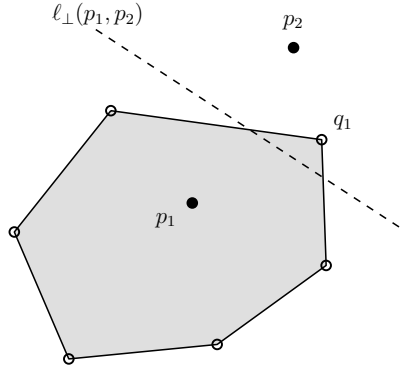, width=.43\linewidth}
    \caption{\CH{Q} intersect the bisector of two data points}
    \label{fig:convex_bisect}
  \end{center}
\end{figure}

\begin{lemma}\label{lemma:geometric-property}
The bisector of two data points intersects the interior of \CH{Q}
if and only if they do not spatially dominate each other.
\end{lemma}

\begin{proof}
 If the bisector of two data points
 intersects the interior of \CH{Q}, then for each of the data points,
 there exists a vertex of \CH{Q} closer to it than the other.
 For example, in Fig.~\ref{fig:convex_bisect}, the bisector of $p_1$ 
 and $p_2$ intersects \CH{Q}, so at least one query point is closer to one
 of each data point than the other.
 Therefore they do not dominate each other.
 If the bisector does not intersect the interior of \CH{Q},
 all the vertices of \CH{Q} (therefore all the query points) are closer
 to one data point than the other. It means one data point spatially 
dominates the other point.
\end{proof}

\noindent As we can determine whether a line intersects the convex hull
or not in $O(\log|\CH{Q}|)$ time by using a binary search technique, 
the dominance test can be done in the same time.

\comment{include algorithm detail. shorten the description in Implementation.}


\begin{lemma}\label{lemma5}
  When \CH{Q} is given, the dominance test for a pair of data points can
  be done in $O(\log|\CH{Q}|)$ time.
\end{lemma}

\subsection{Bounding the Number of Dominance Test}
\label{subsubsec:bound_check} To make the algorithm faster, we
reduce the number of dominance tests. 
Toward the goal, for some vertex $q$ of \CH{Q}, we keep 
the sorted list $\A$ of all the data points in
the ascending order of distance from $q$. 
With this list, we can determine that, if a data point $p_1$ is
located before $p_2$ in \A, then $p_2$ does not spatially dominate
$p_1$ using Lemma~\ref{lemma1}.
Therefore, together with Lemma~\ref{lemma4}, it is sufficient to
perform the dominance test on $p$ only with the spatial skyline
points that are located before $p$ in \A, as we formally state below.

\begin{lemma}\label{lemma6}
  For a data point $p$, if we have the set of all the spatial skyline points
  located before $p$ in \A, we can determine whether $p$ is a spatial skyline
  or not by $O(|S|)$ dominance tests.
\end{lemma}

If there are two data points with the same distance from $q$,
we can break the tie by computing the distances from another vertex
of \CH{Q}. Since no two points have the same distance from three vertices of
\CH{Q}, we only need to do this at most three times.

We now present our algorithm for retrieving all the
spatial skylines. As we can see, the algorithm is surprisingly simple 
and easy to follow.

\begin{algorithm}{SpatialSkyline}{
    \label{algo:exact}
    \qinput{$P, Q$} \qoutput{$S$} }
    initialize the array \A\ and the list $S$\\
    compute the \CH{Q}\label{exact:CH}\\
    \A\qlet the distances from $q_1\in Q$ to every data point\label{exact:dist}\\
    sort \A~in ascending order\label{exact:sort}\\
    \qfor $i \qlet 0$ \qto $|P|-1$ \label{exact:for}\\
    \qdo
    \qif \A[i] is not spatially dominated by $S$ \label{exact:dominance}\\
    \qthen insert \A[i] to $S$ \qfi\qrof\label{exact:rof}\\
    \qreturn $S$
\end{algorithm}

We now analyze the time complexity of~\ref{algo:exact}. 
In line \ref{exact:CH}, the convex hull can be constructed in $O(|Q|\log|Q|)$ 
time~\cite{CG}. Line~\ref{exact:dist} takes $O(|P|)$ time and sorting in 
line~\ref{exact:sort} can be done in $O(|P|\log|P|)$ time.  
In line~\ref{exact:dominance}, we perform the dominance test $O(|S|)$ times, 
each of which takes $O(\log|\CH{Q}|)$ time.
As the \textbf{for} loop in lines from~\ref{exact:for} to~\ref{exact:rof} repeats
$|P|$ times, the entire loop takes $O(|P||S||\log|\CH{Q}|)$ time.
Since $|Q|<|P|$ in most realistic skyline models, the total
time complexity is $O(|P|(|S|\log|\CH{Q}|+\log|P|))$.

\subsection{Bypassing Dominance Tests using the Voronoi Diagram}
In this section, we discuss how we can further reduce dominance tests
by identifying a subset of skyline results,
which we call \emph{seed skylines}, that can be identified
as skyline points 
with no dominant test. That is, before we perform the
algorithm~\ref{algo:exact}, we can quickly retrieve this \emph{seed skylines} to
improve
the performance of the algorithm
dramatically, by bypassing dominance tests on these skylines.

To achieve this goal, we first discuss a relationship of the Voronoi
diagram $\VD{P}$ of a dataset $P$ and $\CH{Q}$.
Theorem~\ref{thm:intersect} describes this relationship between $\VD{P}$
and $\CH{Q}$.

\begin{theorem}[Seed Skyline]\label{thm:intersect}
For given a set $P$ of data points and a set $Q$ of query points,
if the Voronoi cell $\VC{p}$ of $p\in P$ intersects with the boundary
of $\CH{Q}$ or $\CH{Q}$ contains $\VC{p}$,
then $p$ is a skyline point~\cite{Sharifzadeh}.
\end{theorem}
\begin{proof}
See the proofs of Theorem 1 and 3 in \cite{Sharifzadeh}.
\end{proof}

We now present an efficient algorithm to identify the seed skylines, as
the starting point to perform the algorithm~\ref{algo:exact} to identify
the rest of the skyline points.

To retrieve seed skylines efficiently, we first find a
Voronoi cell that contains a vertex of \CH{Q} by using typical point
location query~\cite{CG} on \VD{P}.  From this Voronoi cell, we follow
the edges of \CH{Q} and find the Voronoi cells that intersect the edges.
Then we find Voronoi cells that lie inside \CH{Q} by traversing the
Delaunay graph~\cite{CG}.  Our enhanced algorithm works as follow. Let
$e_i=(q_i,q_{i+1})$ denote the $i$-th edge along the boundary of \CH{Q}.

\begin{algorithm}{SeedSkyline}{
    \label{algo:approx}
    \qinput{$P$, $Q$} \qoutput{$S_{seed}$} }

    initialize $S_{seed}$\\
    compute $\CH{Q}$ and $\VD{P}$\label{approx:CHVD}\\
    find a Voronoi cell \VC{p} containing $q_0$\label{approx:firtcell}\\
    \qfor $i\qlet 0$ \qto $|\CH{Q}|-1$\\
        find all the Voronoi cells $\VC{p}$
        intersecting $e_i$ and insert $p$ to $S_{seed}$\label{approx:intersect}
    \qrof\\
    find all the Voronoi cells $\VC{p}$ lying in $\CH{Q}$
    by traversing Delaunay graph and insert $p$ to $S_{seed}$ \label{approx:delaunay}\\
    \qreturn $S_{seed}$
\end{algorithm}

Note that, we can compute \CH{Q} and \VD{P} in $O(|Q|\log|Q|)$ time 
and in $O(|P|\log|P|)$ time (line~\ref{approx:CHVD}),
respectively, and locate the Voronoi
cell $\VC{p}$ containing the query point $q_0$ in $O(\log|P|)$ time by
point location query on $\VD{P}$ (line~\ref{approx:firtcell}).

\begin{figure}[ht]
  \begin{center}
    \epsfig{file=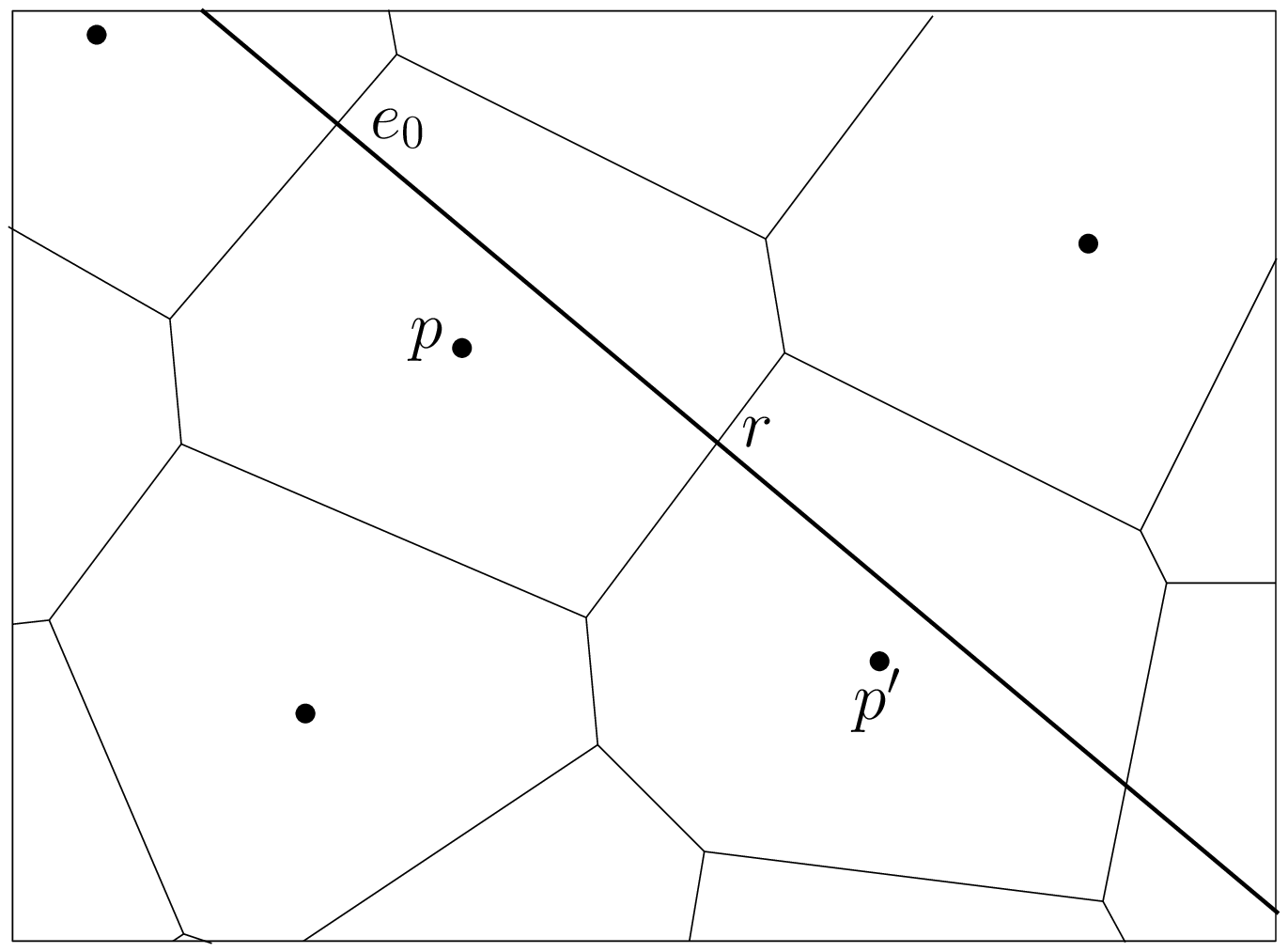, width=.45\linewidth}
    \caption{Two Voronoi cells share the intersection }\label{fig:approx}
  \end{center}
\end{figure}

To find all the Voronoi cells intersecting an edge $e_0=(q_0,q_1)$ in
(line~\ref{approx:intersect}), we follow the procedure below (also
illustrated in Fig \ref{fig:approx}).  We first compute the intersection
$r$ of $e_0$ with the boundary of $\VC{p}$, which can be done in time
$O(\log |P|)$ using binary search because $\VC{p}$ is a convex polygon
and since we store its edges sorted along the boundary, as we will
discuss more later in Section \ref{subsec:vd}. Because $r$ lies on a
boundary edge shared by two neighboring Voronoi cells, we can get the
pointer to the neighboring Voronoi cell $\VC{p'}$ in constant time from
the Delaunay graph. We repeat this until we reach the other endpoint
$q_1$. Then we proceed to the next convex hull edge $e_1=(q_1,q_2)$ and
repeat the above process until we find all the Voronoi cells
intersecting the boundary of \CH{Q}.

Note that a Voronoi cell may contain an edge of $\CH{Q}$ in its interior
or intersect several edges of \CH{Q} -- the number of the intersection
tests is thus bounded by the larger of $O(|S|)$ and $O(|\CH{Q}|)$, \ie,
at most $O(|S|+|\CH{Q}|)$.  Combining the number and cost of
intersection tests, the overall worst-case time complexity becomes
$O((|S|+|\CH{Q}|)\log|P|)$.  Traversing Delaunay graph can be done in
$O(|S|)$ time (line~\ref{approx:delaunay}).  Therefore the total time
complexity of \ref{algo:approx} is $O((|S|+|\CH{Q}|)\log|P|)$ if
$\CH{Q}$ and $\VD{P}$ are given.

\comment{add an algorithm: combination of SpatialSkyline and SeedSkyline, and describe it}

By combining the algorithms~\ref{algo:exact} and~\ref{algo:approx}, 
we can retrieve all spatial skyline points more efficiently than by 
\ref{algo:exact} alone.
Instead of testing dominance for all data points we can find seed 
skylines using~\ref{algo:approx}, and then find the other skylines  
using~\ref{algo:exact}. 
We present the combined algorithm \ref{algo:enhanced} from this idea 
as follows.

\begin{algorithm}{EnhancedSpatialSkyline}{
    \label{algo:enhanced}
    \qinput{$P, Q$} \qoutput{$S$} }

    initialize the array \A\ and the list $S$\\
    compute the \CH{Q}\label{exact:CH}\\
	  $S$ \qlet \ref{algo:approx}$(P,Q)$	\\

    \A\qlet the distances from $q_1\in Q$ to every data point\label{exact:dist}\\
    sort \A~in ascending order\label{exact:sort}\\
    \qfor $i \qlet 0$ \qto $|P|-1$ \label{exact:for}\\
    \qdo
    \qif \A[i] is not in $S$\\
    \qthen\qif \A[i] is not spatially dominated by $S$ \label{exact:dominance}\\
    \qthen insert \A[i] to $S$ \qfi\qfi\qrof\label{exact:rof}\\
    \qreturn $S$
  \end{algorithm}  The asymptotic time complexity of
  \ref{algo:enhanced} is the same as that of
  \ref{algo:exact}. In practice, however, by bypassing the dominance test 
  for seed skylines, it shows better performance than 
  \ref{algo:exact}.

\section{Implementation}\label{sec:imp}

In this section, we discuss the details of our implementation of the
proposed algorithms, including how to compute and store the Voronoi
diagram (Section~\ref{subsec:vd}) and the query convex hull 
(Section~\ref{subsec:ch}) to optimize the implementation of our 
proposed algorithms.

\subsection{Voronoi Diagrams}\label{subsec:vd}

First, we discuss how we construct the Voronoi diagram and the Delaunay
graph of the data points.  As both are extensively studied structures,
many algorithms and codes are available, including `Qhull'
\cite{qhull} which we adopt for our implementation.

However, it is challenging to store the resulting diagram and graph in
such a way that the spatial skyline query computation can be optimized.
Toward the goal, we store the Voronoi cells and Delaunay graph edges as
follows.

\begin{itemize}
\item {\bf cells:} As each Voronoi cell is a convex region, we take
advantage of this convexity and store the vertices of each cell in
increasing angular order from one point, which preserves the
adjacency of vertex pairs in the cell.

\item {\bf edges:}
Every edge of Voronoi cell is shared by a neighboring Voronoi cell. To
represent the Delaunay graph, for each edge $\overline{v_i v_{i+1}}$, from a
vertex $v_i$ of a Voronoi cell, we need to store the pointer to the
neighboring cell sharing the edge.
\end{itemize}

Using this structure, we can exploit the convexity of a Voronoi region and
the Delaunay graph discussed above, by reading only one Voronoi cell block
from file.  To find a specific Voronoi cell block, we maintain a file
pointer for each Voronoi cell block.

\subsection{Convex Hull}\label{subsec:ch}

To compute the convex hull \CH{Q},
we use the \emph{Graham's scan algorithm} \cite{CG}.  By using binary
search technique, the dominance test can be done in $O(\log|\CH{Q}|)$
time, as discussed in Lemma \ref{lemma5}.  We implement the test as
follows.

Remind that we denote the bisector of two data points, $p_1$ and $p_2$,
by $\ell_\perp(p_1, p_2)$.  As discussed in Section
\ref{subsubsec:efficient_check}, we can determine the dominance of two
data points by testing whether $\ell_\perp(p_1,p_2)$ intersects \CH{Q}
or not.  If $\ell_\perp(p_1,p_2)$ intersects \CH{Q}, at least one vertex
of the upper chain of \CH{Q} lies above $\ell_\perp(p_1,p_2)$, and one
vertex of the lower chain of \CH{Q} lies below $\ell_\perp(p_1,p_2)$
(See Fig.~\ref{fig:convex_bisect}). Let $e_i$ and $e_{i+1}$ be two edges
of the upper chain sharing a vertex $q_i$ such that
$\ell_\perp(p_1,p_2)$ has a slope in between the maximum and the minimum
of the slopes of $e_i$ and $e_{i+1}$.  If $\ell_\perp(p_1,p_2)$
intersects \CH{Q}, then $q_i$ lies strictly above $\ell_\perp(p_1,p_2)$
by convexity of \CH{Q}.  We can use a similar argument for the lower
chain of \CH{Q}.  Because the upper and the lower chain of \CH{Q} is
sorted in the increasing order of the slopes of edges, we can find these
two vertices by binary search on the slopes of edges.  After finding
these two vertices in $O(\log|\CH{Q}|)$, we can determine the dominance
in constant time.  When $\CH{Q}$ is small, a linear search may
outperform binary search, and we use linear search in this case.

\subsection{\vstwo}\label{subsec:vs2}

As a baseline to compare with our proposed algorithm, we use \vstwo\
proposed in \cite{Sharifzadeh}.  As the authors could not provide the
code, we implement the algorithm using the same implementation of
\rtree~\cite{rstartree} and the Voronoi diagram we used to implement our
proposed algorithm, to ensure the fairness in empirical comparison.

For constructing the convex hull, we share the same implementation used
for our proposed algorithms, except that, to accommodate the dominance
test of complexity $O(|\CH{Q}|)$ discussed in \cite{Sharifzadeh}, we
use linear scan.

In our implementation, \rtree\ is used to find the closest point to one
query point. The leaves of a \rtree\ index contain Voronoi cells which
are packed by MBRs for each, such that we can easily obtain candidate
Voronoi cells containing a query point. 

However, as shown in Section~\ref{subsec:naive}, \vstwo\ may fail to
find all the spatial skyline points in some cases. Our implementation of
\vstwo\ is revamped to eliminate these cases. Specifically, we remove
one condition. For some $p\in P$, if all its Voronoi neighbors and all
their Voronoi neighbors are spatially dominated by other points, then
original \vstwo\ does not test $p\in P$, but we implement \vstwo\ to
test this point for finding all skyline points.

\begin{figure}[ht]
    \centering
    \epsfig{file=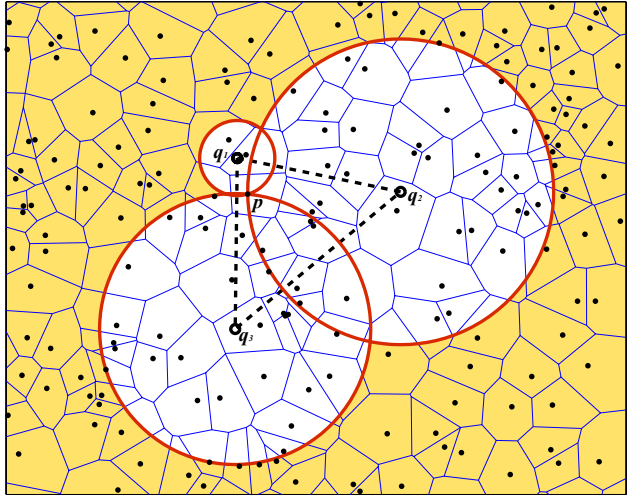, width=.68\linewidth}
    \caption{Dominating region of a skyline object}
    \label{fig:domregion}
\end{figure}

\subsection{Enhanced Spatial Skyline (\esky)}\label{subsec:es}
Our enhanced algorithm works as follows.
We compute the Voronoi diagram and the Delaunay graph of the data points, 
and store them in the form of the file mentioned in
Section~\ref{subsec:vd}.  To find the point closest to one query point,
\rtree\ is used.  Then \esky\ computes the Voronoi cells intersecting
the boundary of the query convex hull and find all the Voronoi cells
lying in the convex hull by traversing the Delaunay graph.  As we only
need to see each Voronoi cell at most once during traversing the
Delaunay graph of the data points, we read it from the file when it is
required and deallocate it from memory after passing it by.

In this process, we restrict the region to search for the rest of the
skylines to the bounding box containing $|Q|$ circles for $|Q|$ query
points (Fig.~\ref{fig:domregion}). More precisely, we set the bounding
box as the intersection of all bounding boxes defined by the skyline
subset found so far. After that, we get a list of the candidates in this
bounding box by using \rtree.  We sort the list in the ascending
order of the candidates' distances to a query point and process them one
by one in this order.  When we find a new skyline point, we reduce the
size of the bounding box by taking the intersection of the current
bounding box with the bounding box of this new skyline point.  During
the process, if some candidate point is not contained in the bounding
box, then we can simply skip the dominance test.

\section{Experiments}\label{sec:exp}

In this section, we report our experiment settings
(Section~\ref{subsec:expsetting}) and evaluation results to validate the
efficiency of our framework (Section \ref{subsec:effi}).  We compared our
algorithm for spatial skylining with \vstwo\ in several aspects.  As
datasets, we used both synthetic datasets and a real dataset of points of
interest (POI) in California.\footnote{Available at
http://www.cs.fsu.edu/\~{}lifeifei/SpatialDataset.htm}

\subsection{Experiment Settings}\label{subsec:expsetting}

{\bf Synthetic dataset:}
A synthetic dataset contains up to one million uniformly distributed random
locations in a 2D space. The space of datasets is limited to a unit space,
\ie, the upper and lower bound of all points are 0 and 1 for each dimension
respectively. More precisely, We used five synthetic datasets with 50K,
100K, 200K, 500K, and 1M uniformly distributed points.

Using synthetic datasets, we investigated the effect of the number of
points in a query $|Q|$, distribution of the points in a query $\sigma$,
and cardinalities of the datasets $|P|$.  Parameters used in the
experiments are summarized in Table~\ref{tab:param}.

\begin{table}[ht]\centering
    \caption{Parameters used for synthetic datasets}
    \begin{tabular*}{\linewidth}{p{.5\linewidth}p{.45\linewidth}}\hline\hline
        Parameter & Setting (\underline{Default}) \\ \hline
        Dimensionality & 2 \\
        Dataset cardinality & 50K, 100K, 200K, \underline{500K}, 1M \\
        Distribution of data points& Independent \\
        The number of points in a query & 5, 10, \underline{15}, 20, 40 \\
        Standard deviation of points in a query & 0.01, 0.02, 0.04, \underline{0.06}, 0.08 \\ \hline\hline
    \end{tabular*}
    \label{tab:param}
\end{table}

Queries were generated through the following steps: (1) We randomly
generate a \emph{center point} then (2) generate the query points,
normally distributed around the center. In particular,
for each dimension, we generate points that are normally distributed, 
with mean as the center point and deviation as user-specified parameter
$\sigma$, which varies among 0.01, 0.02, 0.04, 0.06, and 0.08 as listed in
Table~\ref{tab:param}.  We generated hundred queries (each consisting of up
to 40 query points) for each setting and measured average response times of
all algorithms.

{\bf POI dataset:}
We also validate our proposed framework using real-life dataset.
In particular, we use a POI dataset, which consists of 104,770 locations of
63 different categories in California.  Fig.~\ref{fig:poi} shows the
characteristics of this POI dataset.

\begin{figure}[h]\centering
	\epsfig{file=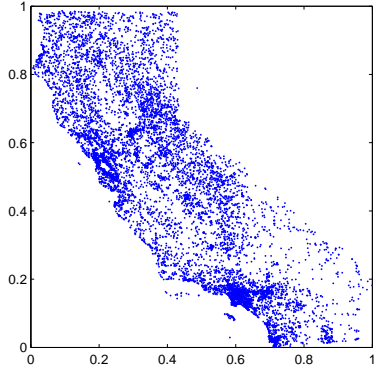, clip=, width=0.4\linewidth}
	\caption{10,000 sampled points from the California's POI dataset}
	\label{fig:poi}
\end{figure}

For this POI dataset, we investigated the effect of $|Q|$ and $\sigma$.
We similarly generated the queries, by randomly picking one data point as
a center point and generating query points to be normally distributed
around the center point, in the same way we generated synthetic points.
The reason why we pick the center point among data points, instead of
generating a random point, is to avoid generating queries to regions with
no data points (such as blank regions in Fig.~\ref{fig:poi}.  We generate
hundred queries for each setting, varying the number of query points in the
range from 5 to 40 and the standard deviation from 0.01 to 0.08, just as in
our synthetic data point generation.

We carry out our experiments on a Pentium IV PC running on Linux with
Pentium IV 3.2GHz CPU and 1GB memory, and all the algorithms were coded in C++.

\subsection{Efficiency}\label{subsec:effi}

We validate the efficiency of our framework, over varying
$|P|$, $|Q|$, and $\sigma$.

\begin{figure}[ht]\centering
	\begin{tabular}{ccc}
		\epsfig{file=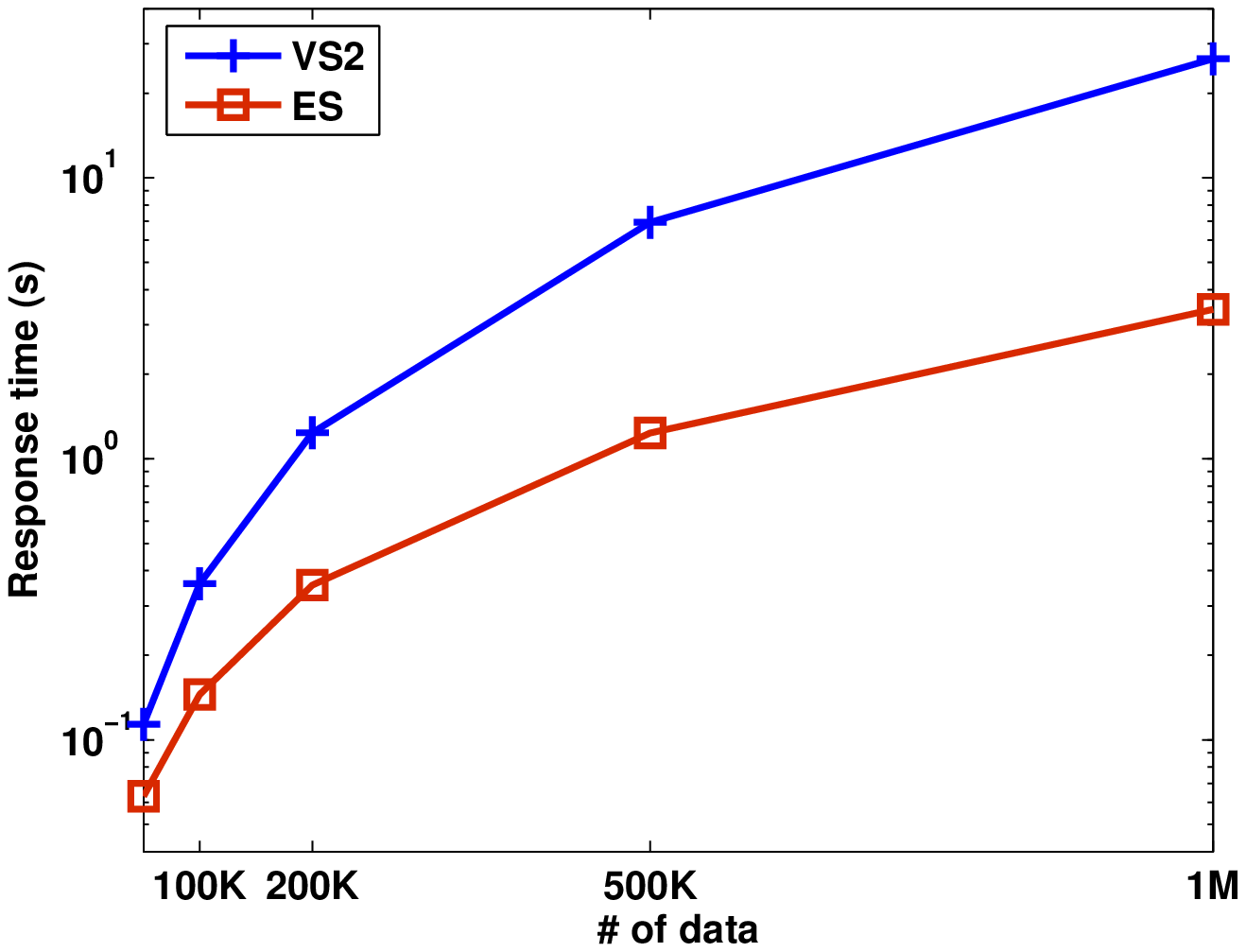, width=.32\linewidth, bb=96 265 480 565} &
		\epsfig{file=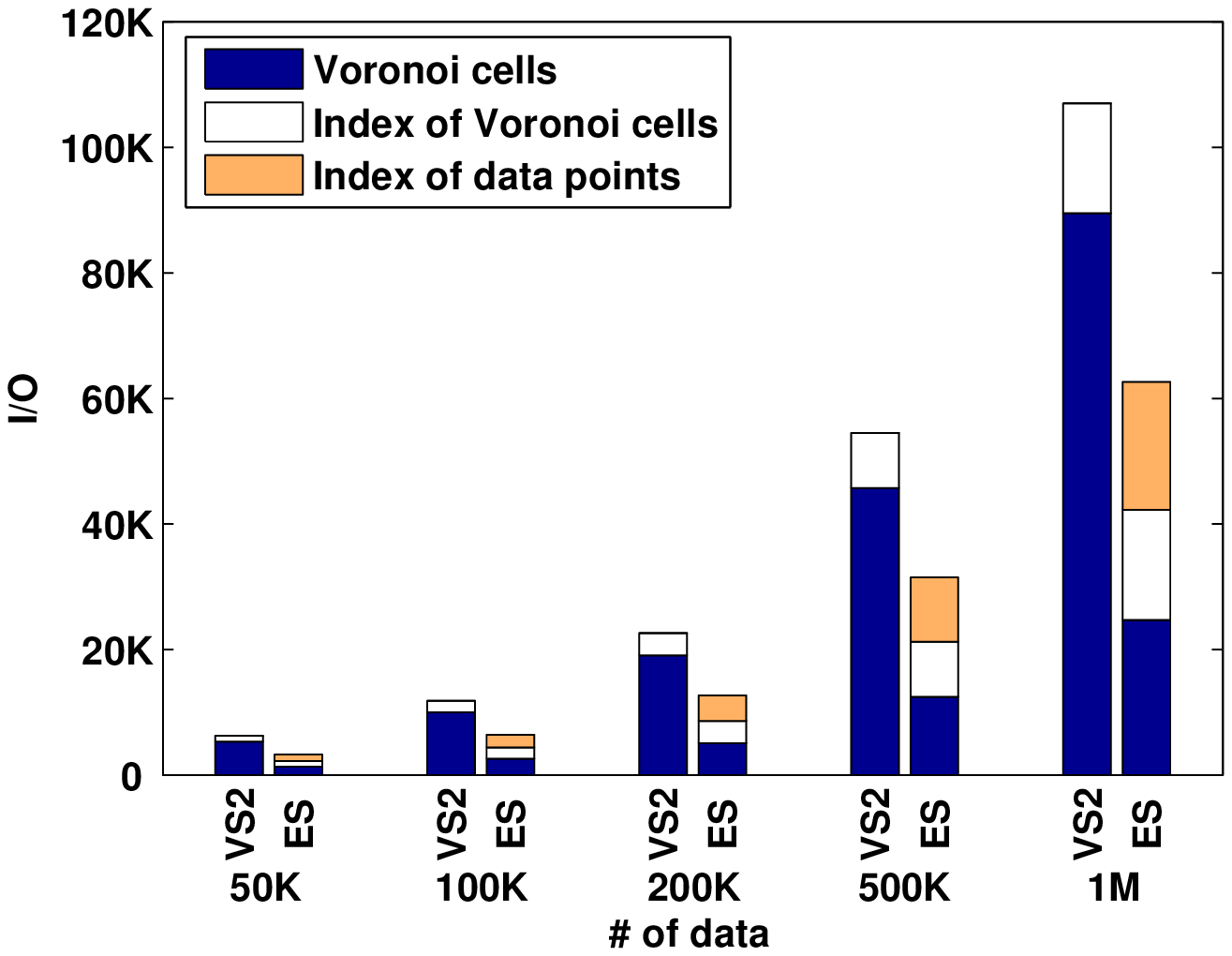, width=.32\linewidth, bb=96 265 480 565} &
		\epsfig{file=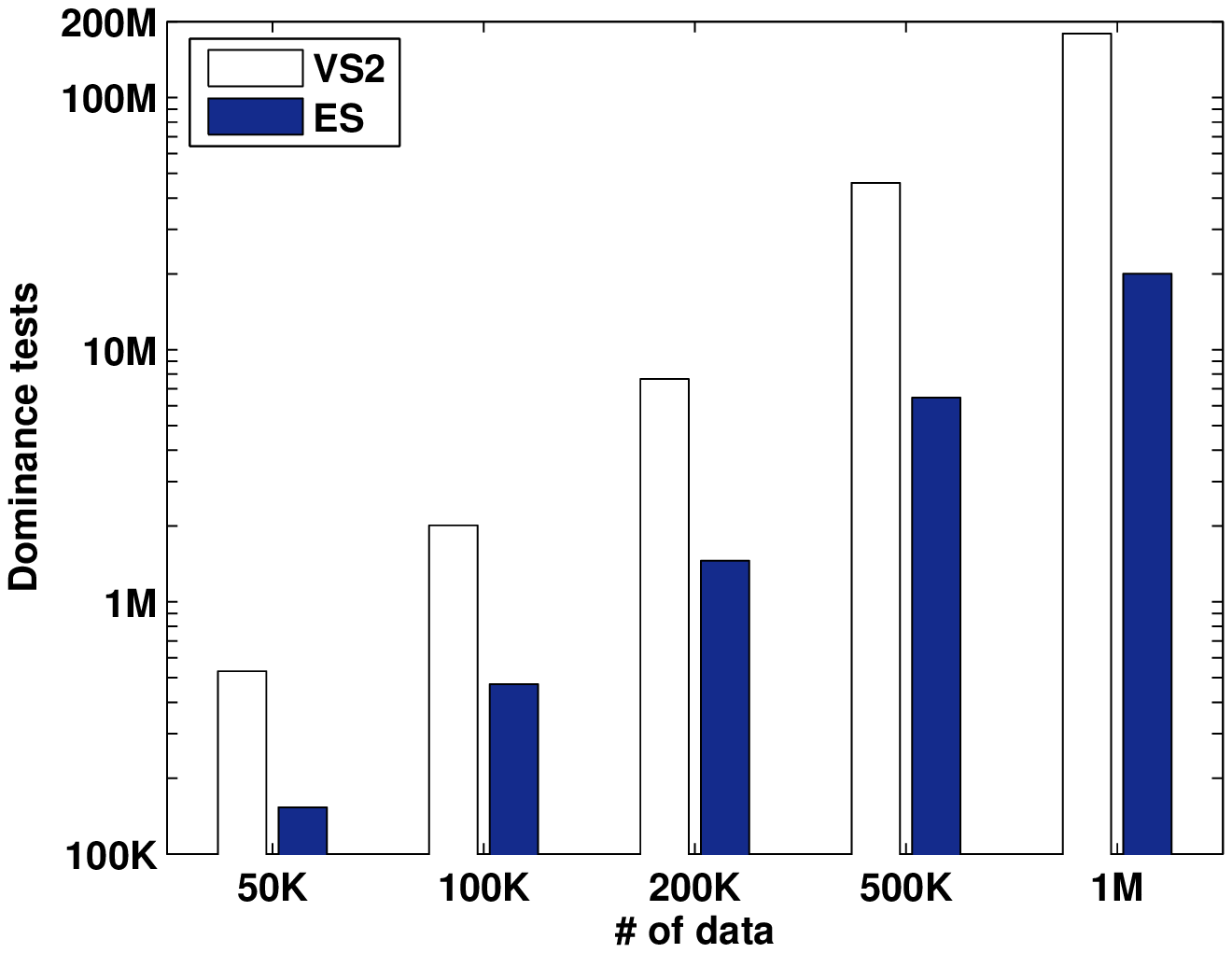, width=.32\linewidth, bb=96 265 480 565} \\
		(a) Response time &
		(b) I/O &
		(c) Dominance tests
	\end{tabular}
	\caption{Effect of the dataset cardinality for synthetic datasets}
	\label{fig:synth:datasize}
\end{figure}

Fig.~\ref{fig:synth:datasize} shows the effect of the dataset cardinality
to response time (Fig.~\ref{fig:synth:datasize}a), 
I/O cost, measured as the number of accessing (reading) Voronoi cells
and \rtree\ nodes, (Fig.~\ref{fig:synth:datasize}b), 
and the number of dominance tests (Fig.~\ref{fig:synth:datasize}c).

From Fig.~\ref{fig:synth:datasize}a, observe that our proposed algorithm
\esky\ outperforms \vstwo\ by an order of magnitude.  Similarly in
Fig.~\ref{fig:synth:datasize}c, \esky\ performs a remarkably smaller number
of dominance tests than \vstwo, by bypassing the dominance tests for the
skylines whose Voronoi cells intersect the boundary of $\CH{Q}$.  
Such saving is more significant between skylines, 
as the number of the dominance tests for
skylines is significantly higher.

Fig.~\ref{fig:synth:datasize}b shows the I/O costs of the three algorithms--
Observe that, three algorithms perform same number of I/Os on the index of
Voronoi cells, because each algorithm only uses the index to find a Voronoi
cell containing a query point.  To find non-seed skylines, \esky\ uses the
index of data points, which incurs less I/Os (random accesses) than \vstwo.
\esky, though the size of each I/O (\rtree\ node) is larger than that of
\vstwo\ (a Voronoi cell), outperforms \vstwo\ by reducing the ``number" of
I/Os, each of which incurs a random access, the cost of which dominates the
overall access cost, in our scenario of performing many random accesses of
smaller size.

\begin{figure}[ht]\centering
	\begin{tabular}{ccc}
		\epsfig{file=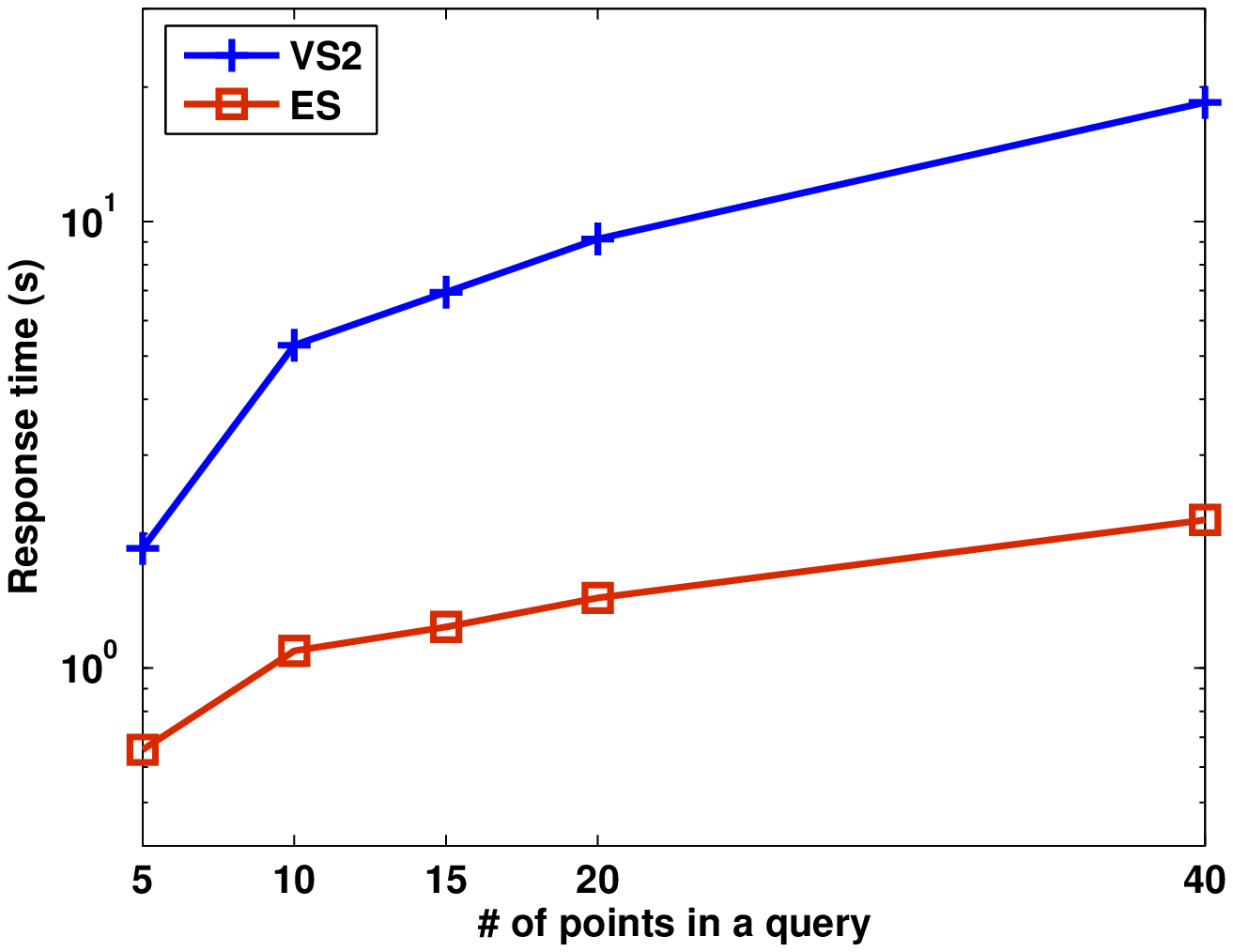, width=.32\linewidth, bb=96 265 480 565} &
		\epsfig{file=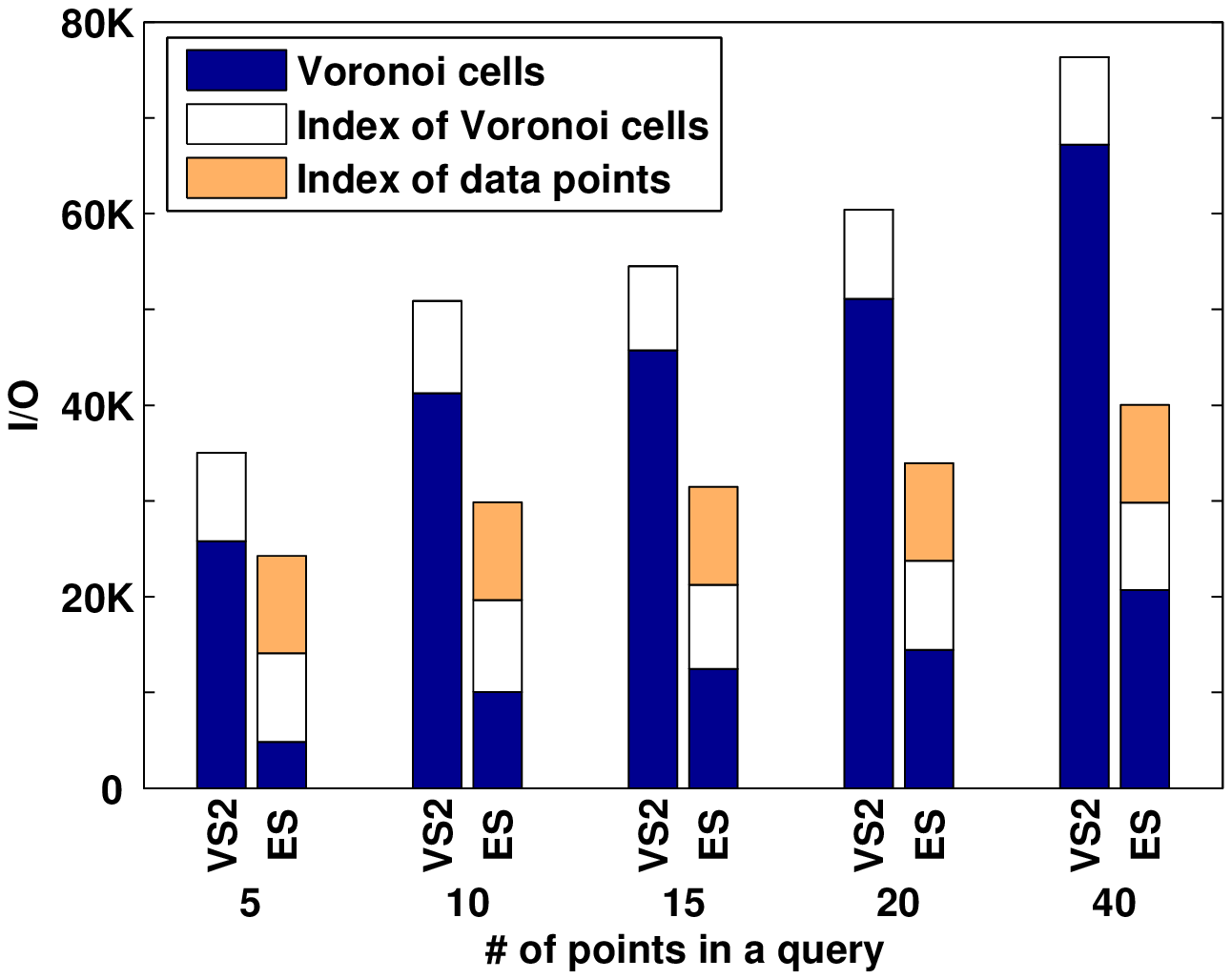, width=.32\linewidth, bb=96 265 480 565} &
		\epsfig{file=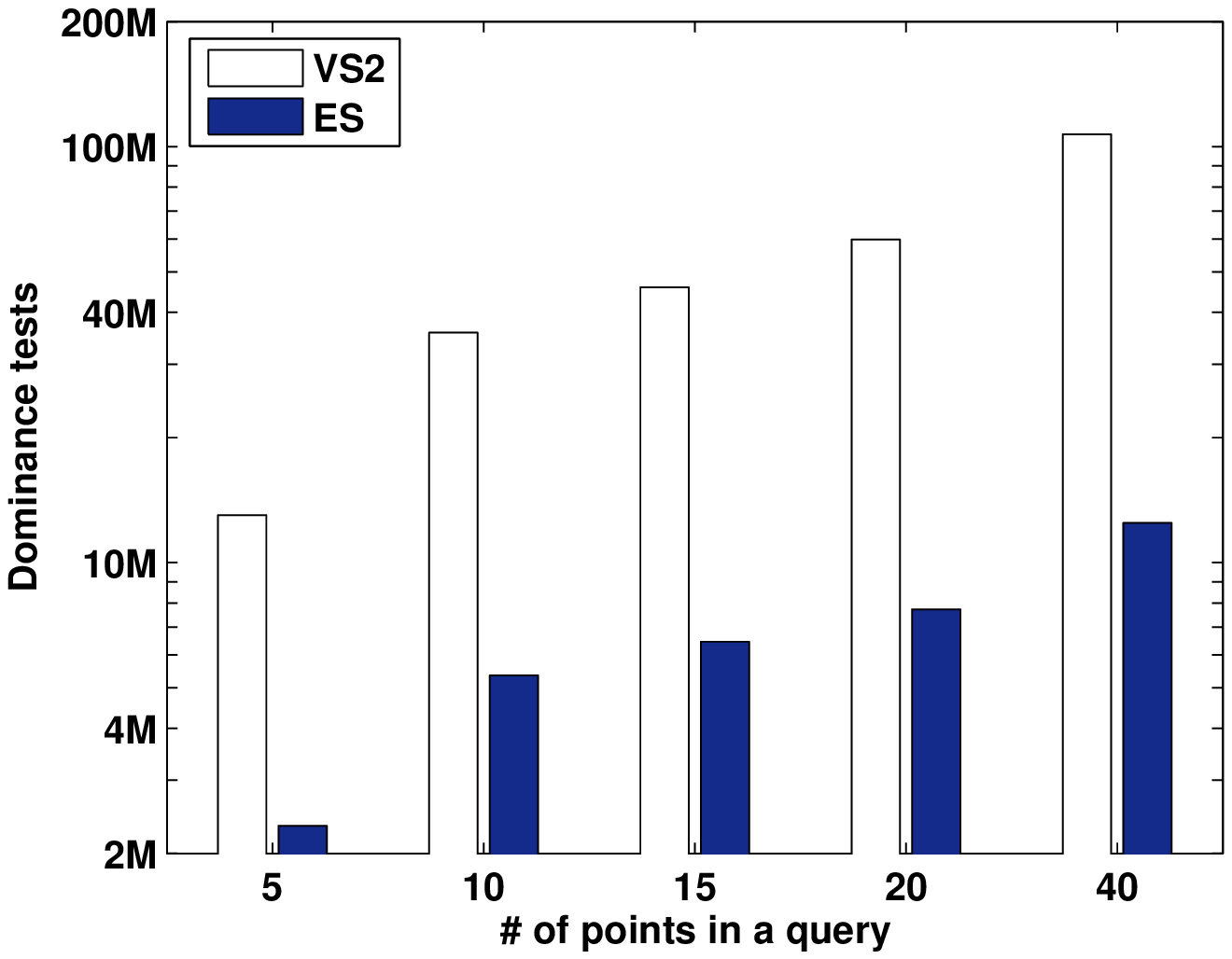, width=.32\linewidth, bb=96 265 480 565} \\
		{(a) Response time} &
		{(b) I/O} &
		{(c) Dominance tests}
	\end{tabular}
	\caption{Effect of the number of query points for synthetic datasets}
	\label{fig:synth:querysize}
\end{figure}

Fig.~\ref{fig:synth:querysize} shows the effect of $|Q|$ to response time,
I/O cost, and the number of dominance tests.  We observe similar trends as
in Fig.~\ref{fig:synth:datasize}, except that the response time and I/Os
scale more gracefully over increasing $|Q|$.  This can be explained by the
fact that all the three algorithms use $\CH{Q}$, instead of using $Q$
itself, the size of which grows much slowly  than that of $Q$.  For
instance, even when $|Q|$ is doubled, the size of convex hull may not
change much, if the deviation $\sigma$ stays the same.

\begin{figure}[ht]\centering
	\begin{tabular}{ccc}
		\epsfig{file=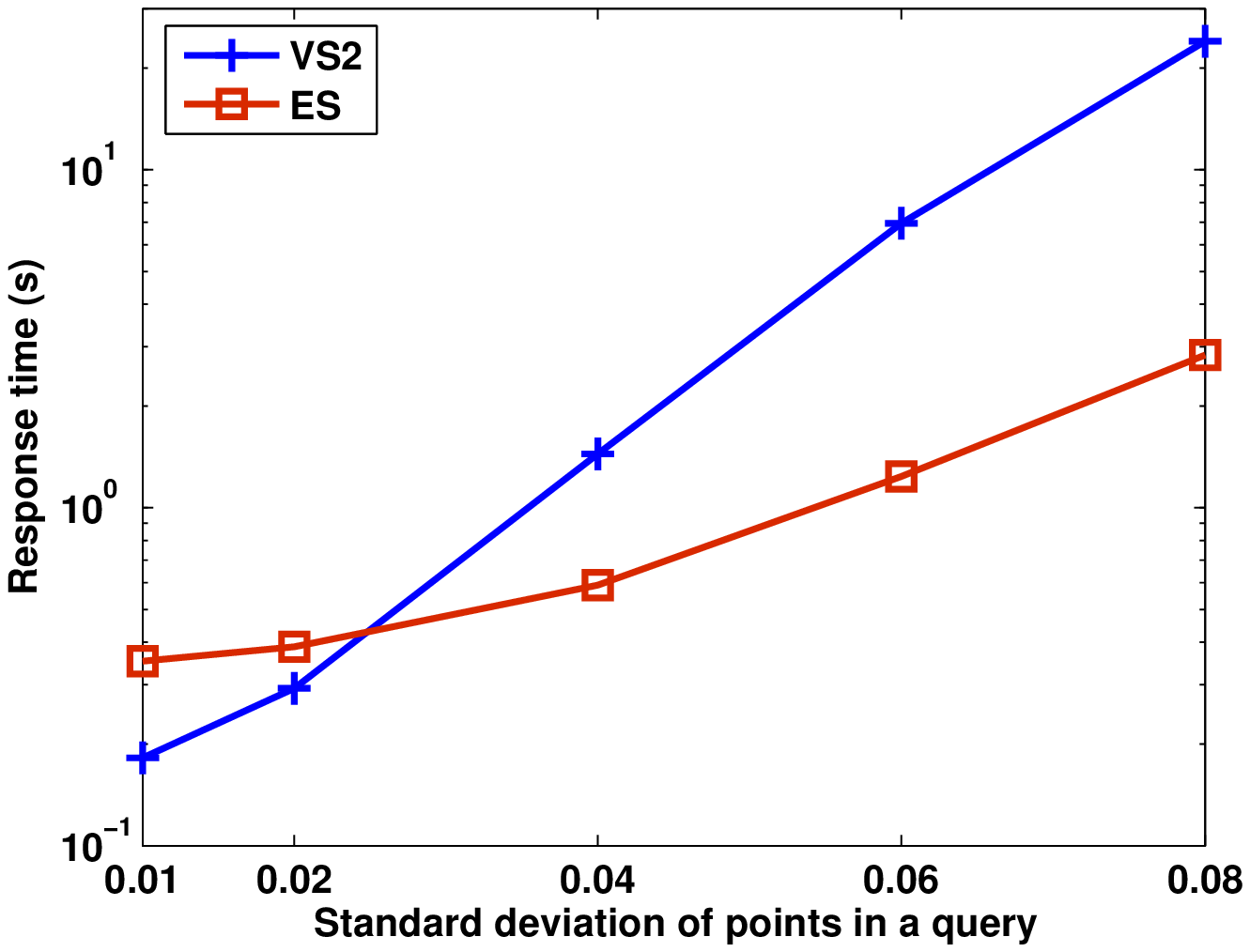, width=.32\linewidth, bb=96 265 480 565} &
		\epsfig{file=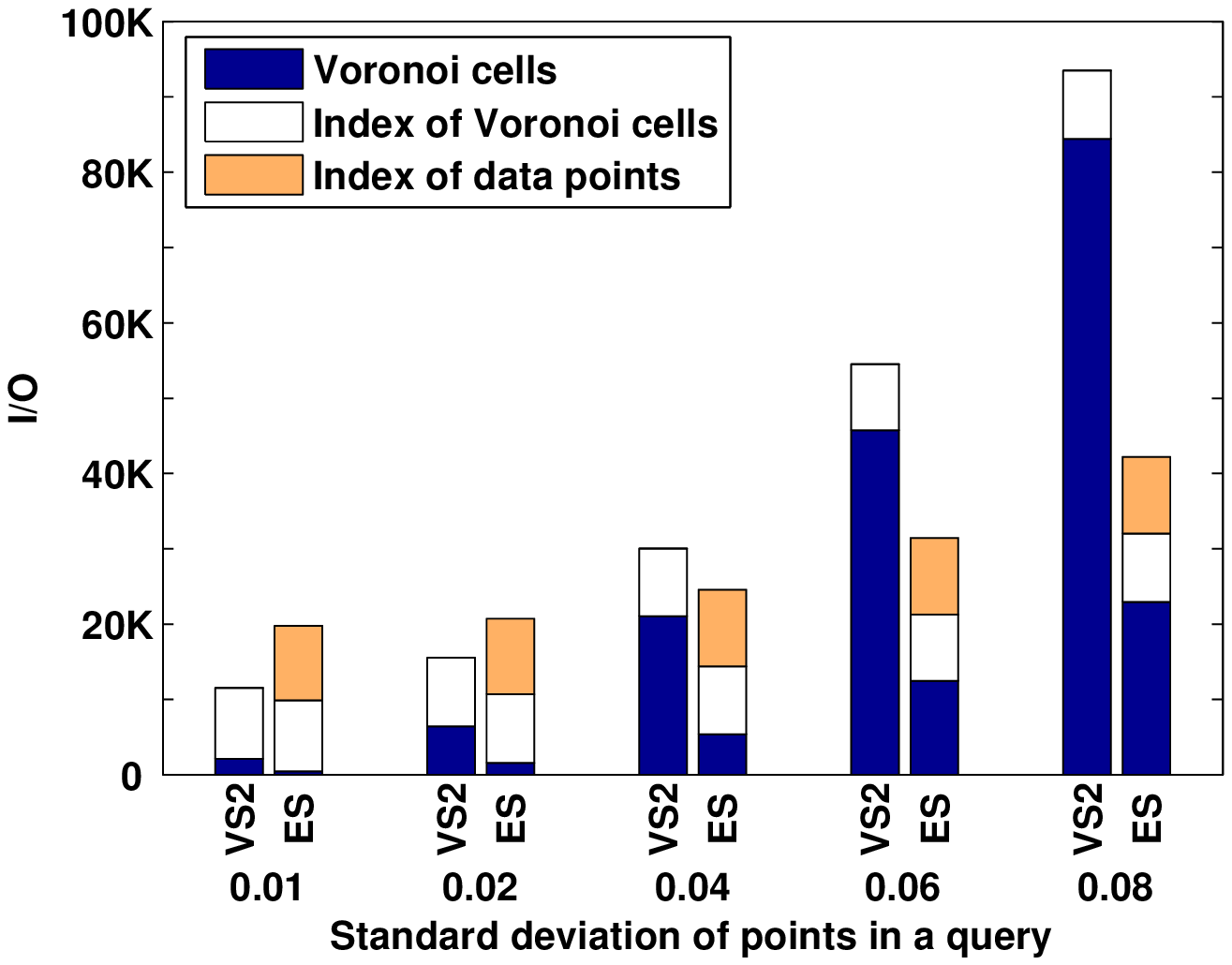, width=.32\linewidth, bb=96 265 480 565} &
		\epsfig{file=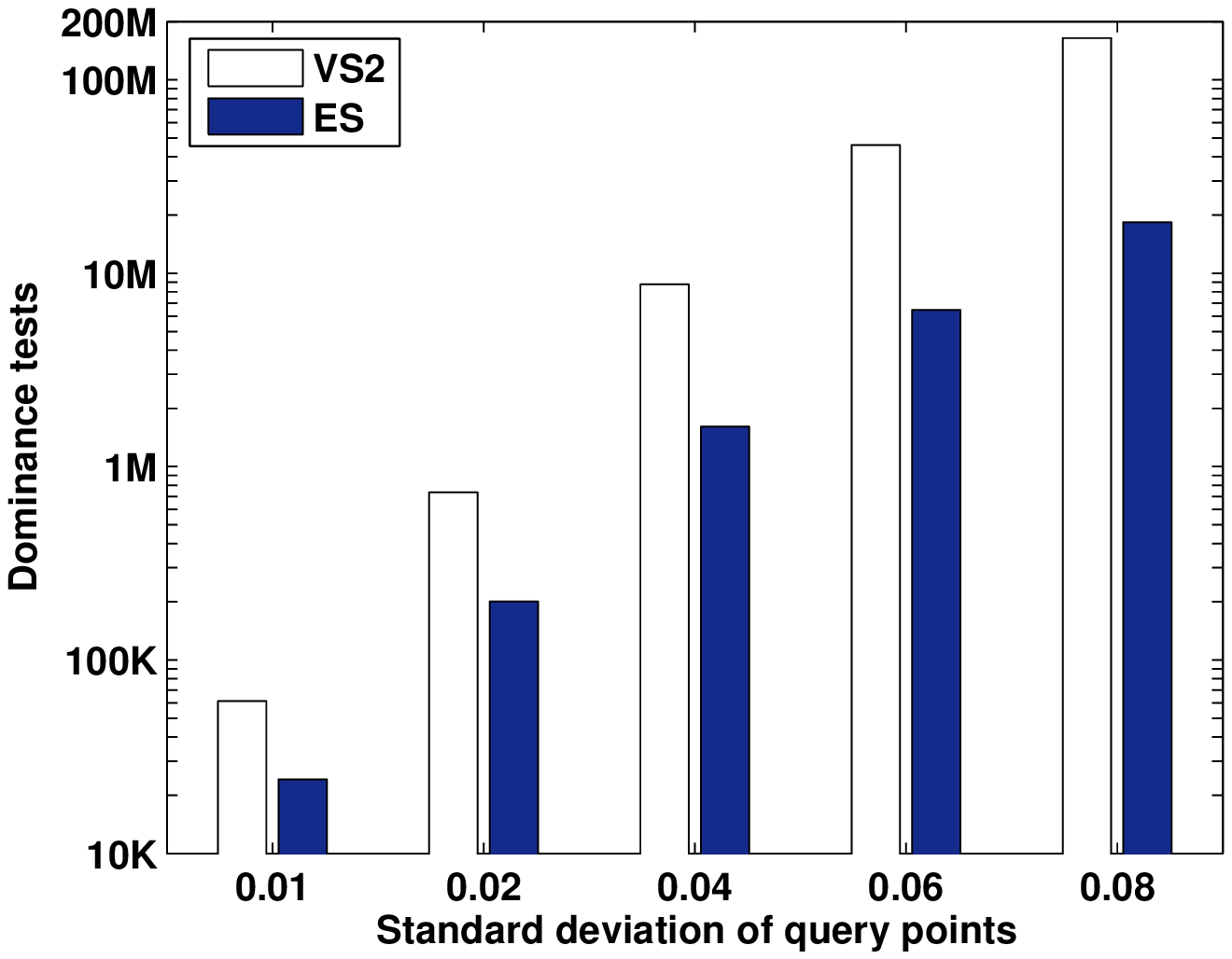, width=.32\linewidth, bb=96 265 480 565} \\
		{(a) Response time} &
		{(b) I/O} &
		{(c) Dominance tests}
	\end{tabular}
	\caption{Effect of $\sigma$ of a query for synthetic datasets}
	\label{fig:synth:sigma}
\end{figure}

Fig.~\ref{fig:synth:sigma} shows the effect of $\sigma$.  Similarly to
prior results, \esky\ significantly outperforms \vstwo\ in terms of
response time, dominance tests, and I/Os while \vstwo\ outperforms our
algorithm when query points are crowded in a very small area.  This
phenomenon can be explained as \esky\ performs more I/Os than \vstwo\ when
the size of $\CH{Q}$ is very small (Fig.~\ref{fig:synth:sigma}b).  However,
\esky\ starts to outperform \vstwo\ as the size of $\CH{Q}$ grows.

The other slight difference to note is that the response times of the
algorithms increase relatively faster as $\sigma$ increases, as the size of
$\CH{Q}$ may increase quadratically as $\sigma$ increases.  For example,
when $\sigma$ changes from $0.04$ to $0.08$ (two-fold), the circle area
containing the points within the 95\% confidence interval increases
four-fold (\ie, quadratic), and also the area of $\CH{Q}$ and the points
inside $\CH{Q}$.  As such points are guaranteed to be skylines, this
observation suggests why the number of skylines increases quadratically as
$\sigma$ increases.

\begin{figure}[ht]\centering
	\begin{tabular}{ccc}
		\epsfig{file=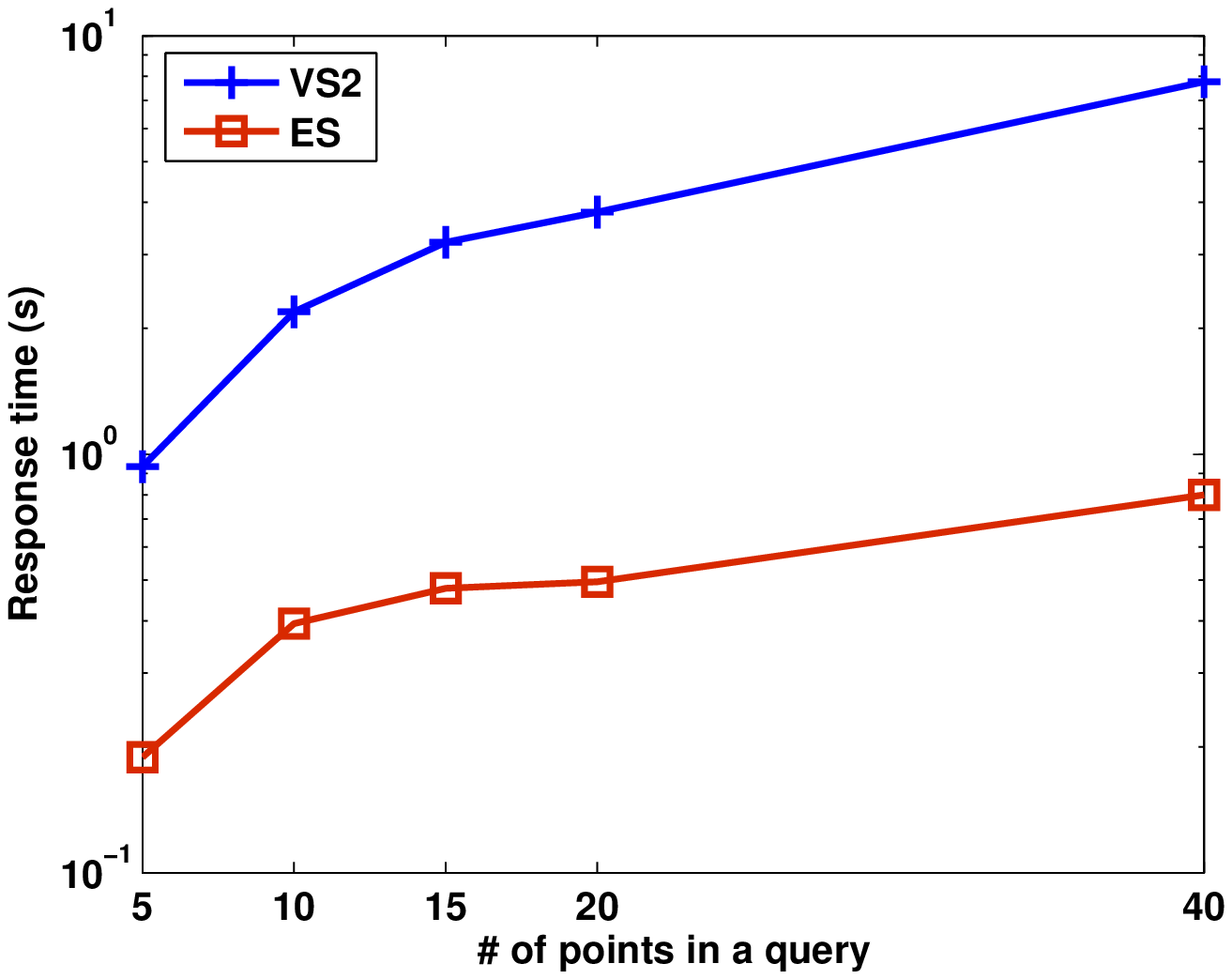, width=.32\linewidth, bb=96 265 480 565} &
		\epsfig{file=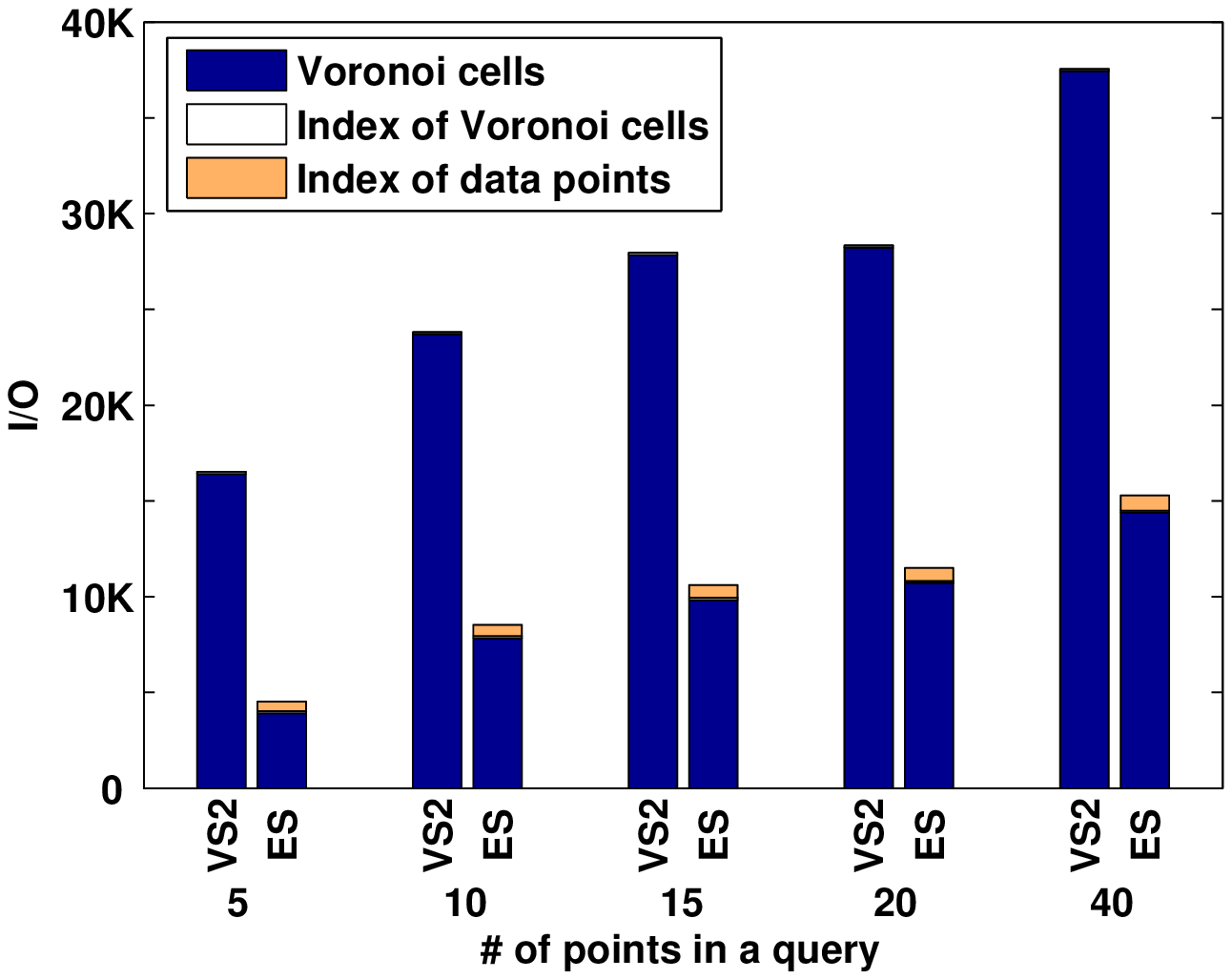, width=.32\linewidth, bb=96 265 480 565} &
		\epsfig{file=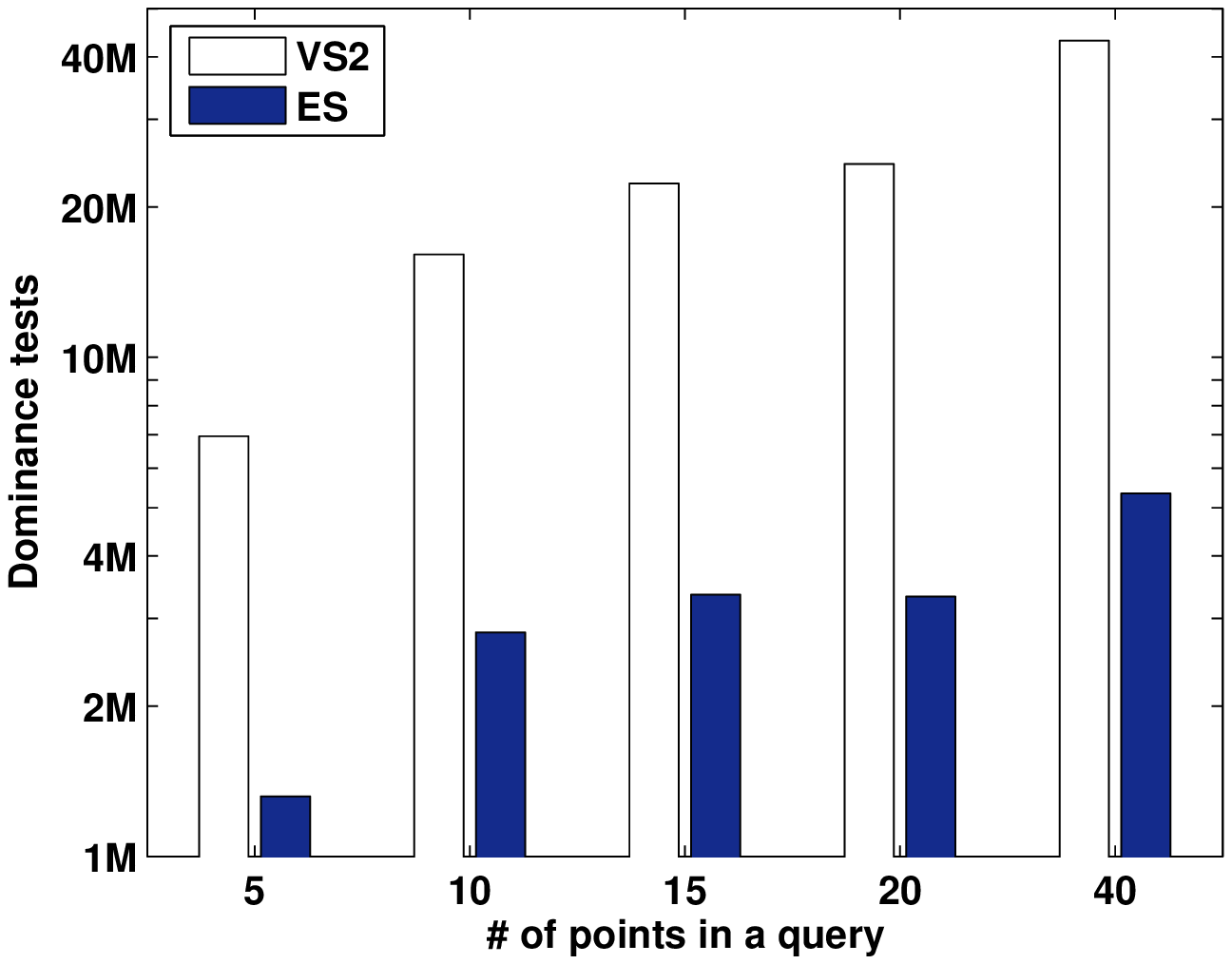, width=.32\linewidth, bb=96 265 480 565} \\
		{(a) Response time} &
		{(b) I/O} &
		{(c) Dominance tests}
	\end{tabular}
	\caption{Effect of the number of query points for the POI dataset}
	\label{fig:real:querysize}
\end{figure}

\begin{figure}[ht]\centering
	\begin{tabular}{ccc}
		\epsfig{file=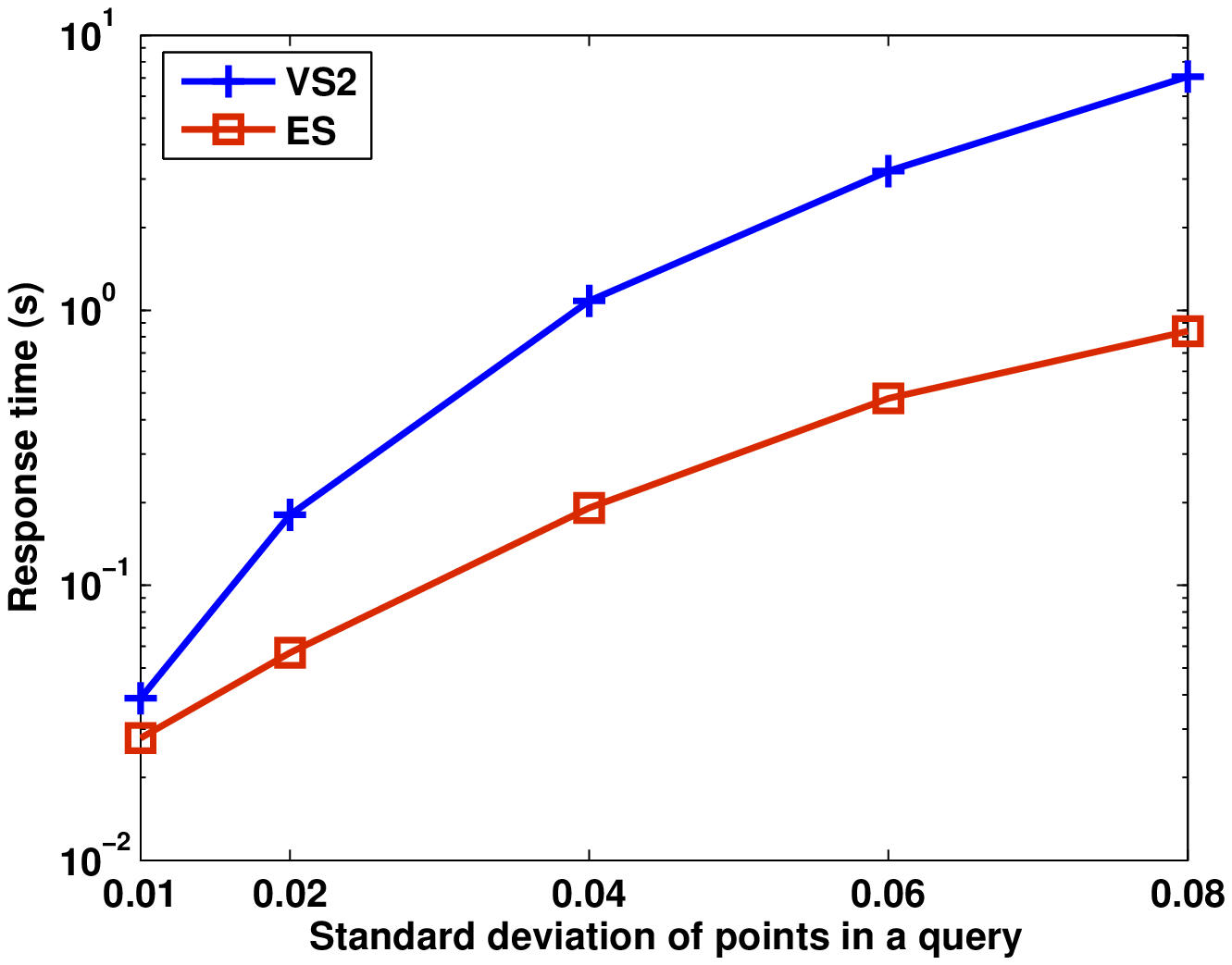, width=.32\linewidth, bb=96 265 480 565} &
		\epsfig{file=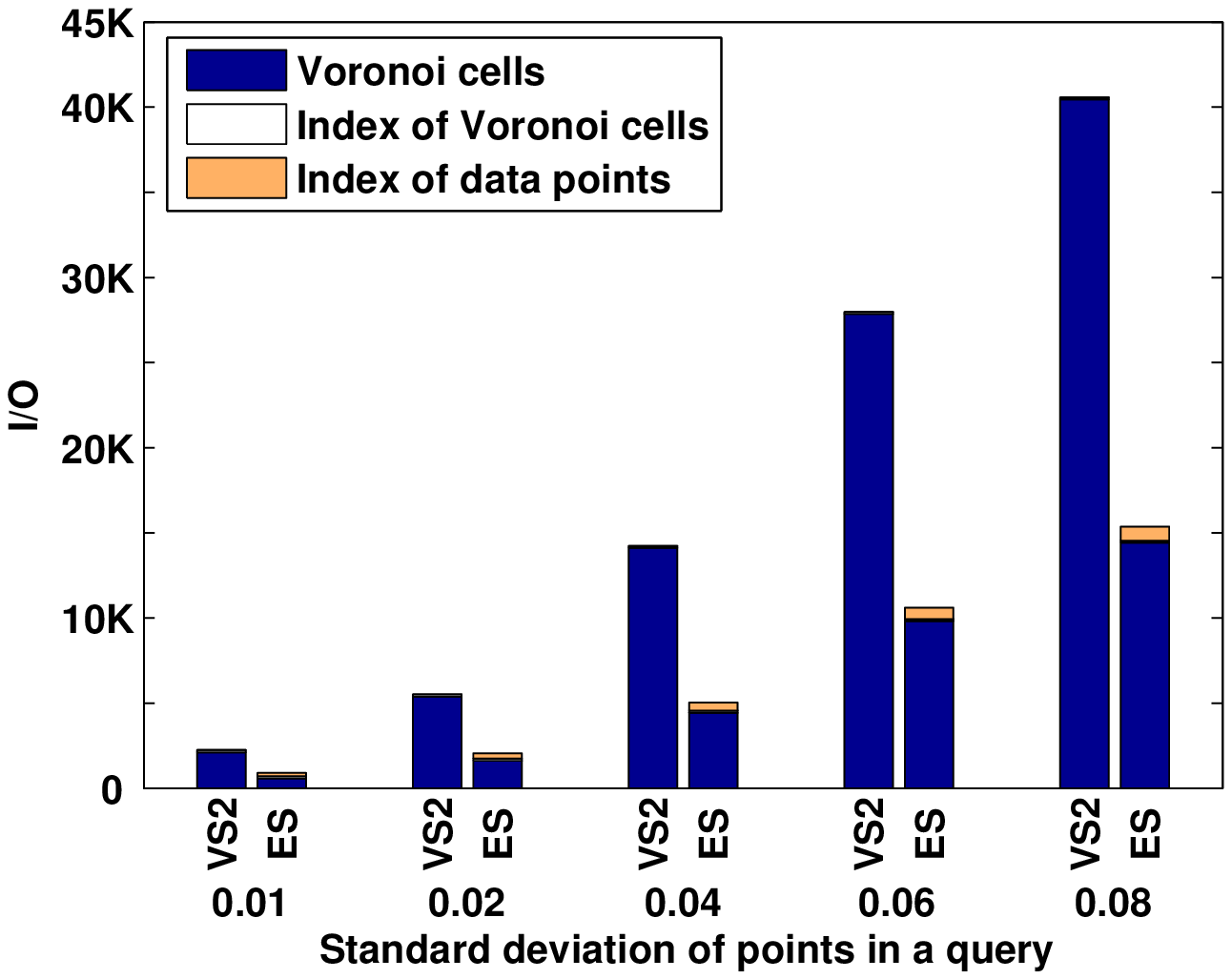, width=.32\linewidth, bb=96 265 480 565} &
		\epsfig{file=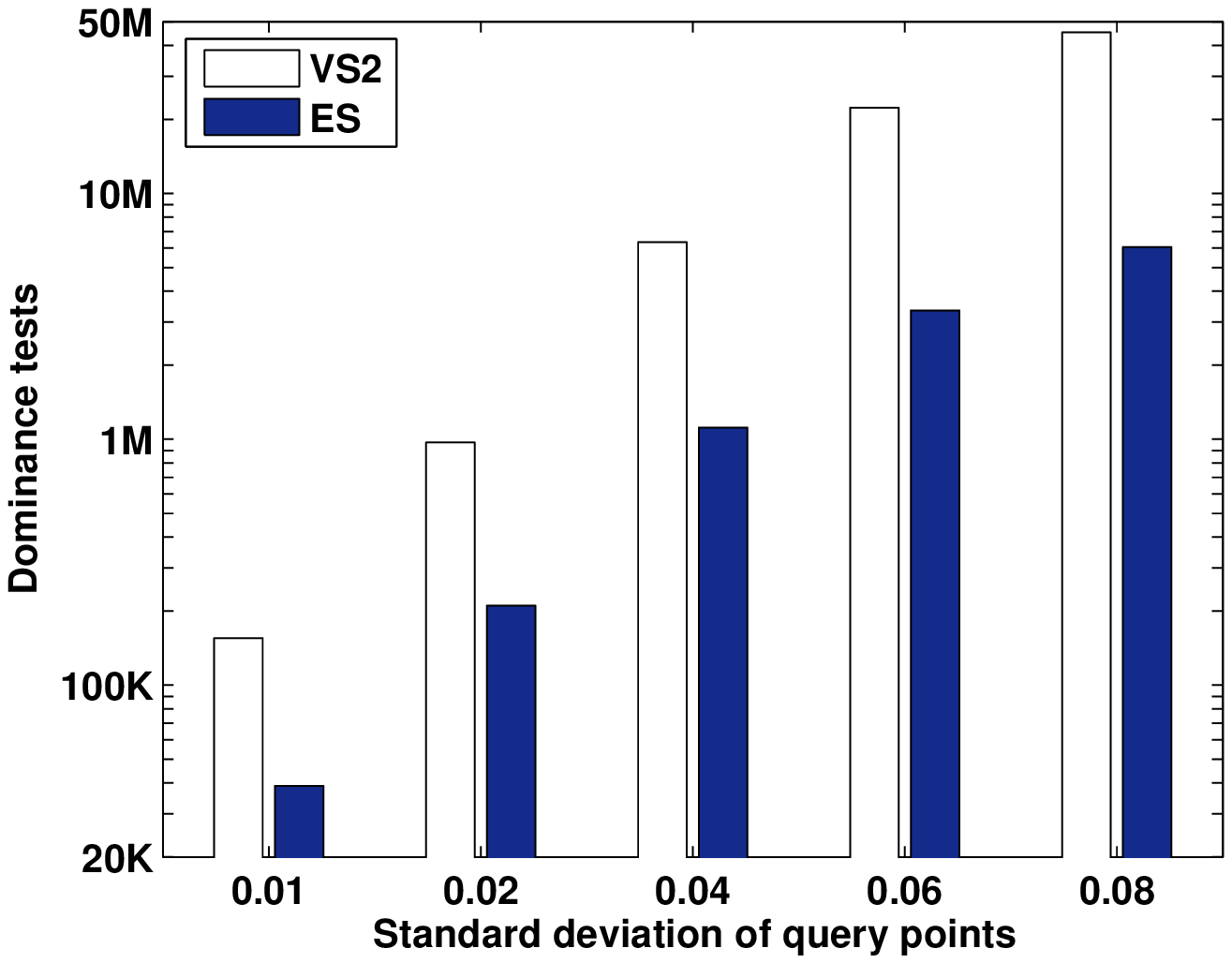, width=.32\linewidth, bb=96 265 480 565} \\
		{(a) Response time} &
		{(b) I/O} &
		{(c) Dominance tests}
	\end{tabular}
	\caption{Effect of $\sigma$ of a query for the POI dataset}
	\label{fig:real:sigma}
\end{figure}

We perform the same sets of experiments on the POI dataset, varying the
size of query and $\sigma$, reported in Fig.~\ref{fig:real:querysize} and
\ref{fig:real:sigma} respectively.  Our observations of these evaluations
are roughly consistent with the corresponding evaluation for synthetic
datasets.  However, in these experiments, I/Os on Voronoi cells are
dominant parts of the I/O cost.  The reason is that, as the cardinality of
the dataset is relatively smaller, the depth of the \rtree\ is also small,
thus incurring less index I/Os.  A similar phenomenon can be observed in
Fig.~\ref{fig:synth:datasize}b, when the dataset cardinality is small (50K).

\section{Conclusion}\label{sec:con}
We have studied spatial skyline query processing and presented an efficient
and correct
exact algorithm. We showed that our algorithm
can identify the correct result in $O(|P|(|S|\log|\CH{Q}|+\log|P|))$
time, while the best known algorithm may fail to compute the correct
result. 
Lastly, we empirically validated our proposed algorithm.

So far we have assumed that the points lie in $2$-dimensional space, and
shown how to efficiently retrieve spatial skyline points using some
geometric structures such as the convex hull and the Voronoi diagram of
points in the plane.  We now turn our attention to higher dimensional
skyline queries. All the definitions, lemmas, and algorithms described in
this paper generalize to higher dimensions: For the set of $n$ points in
$d$-dimensional space, the Voronoi diagram of them has $\Theta(n^{\lceil
d/2\rceil})$ combinatorial complexity~\cite{Klee80} and can be computed
in $O(n\log n+n^{\lceil d/2\rceil})$
time~\cite{Chazelle91,Clarkson89,Seidel91}. The convex hull of those
points has $\Theta(n^{\lfloor d/2\rfloor})$ combinatorial complexity (by
the so-called \emph{Upper Bound Theorem}) and can be computed in
$\Theta(n^{\lfloor d/2\rfloor})$ expected time~\cite{CG}. The dominance
test, the intersection query of a line with a convex polygon used in
Section~\ref{subsubsec:efficient_check}, can be generalized for higher
dimensions, as intersection query of a hyperplane with a convex
polyhedron in higher dimensions. Similarly, the intersection of an edge
with the Voronoi diagram can also be generalized as the intersection of a
$d-1$-face with the Voronoi diagram in $d$-dimensional space.

For future work, we will study how our algorithms can be extended to
support queries over urban road networks with additional constraints.

\bibliographystyle{splncs}
\bibliography{ssq}

\begin{thebibliography}{10}

\bibitem{Kung75}
Kung, H.T., Luccio, F., Preparata, F.:
\newblock On finding the maxima of a set of vectors.
\newblock In: Journal of the Association for Computing Machinery

\bibitem{borzsonyi01}
B{\"o}rzs{\"o}nyi, S., Kossmann, D., Stocker, K.:
\newblock The skyline operator.
\newblock In: ICDE '01: Proc. of the 17th International Conference on Data
  Engineering. (2001)  421

\bibitem{Tan01}
Tan, K., Eng, P., Ooi, B.C.:
\newblock Efficient progressive skyline computation.
\newblock In: VLDB '01: Proc. of the 27th International Conference on Very
  Large Data Bases. (2001)  301--310

\bibitem{Papadias03}
Papadias, D., Tao, Y., Fu, G., Seeger, B.:
\newblock An optimal and progressive algorithm for skyline queries.
\newblock In: SIGMOD '03: Proc. of the 2003 ACM SIGMOD International Conference
  on Management of Data. (2003)  467--478

\bibitem{Chomicki03}
Chomicki, J., Godfery, P., Gryz, J., Liang, D.:
\newblock Skyline with presorting.
\newblock In: ICDE '07: Proc. of the 23rd International Conference on Data
  Engineering. (2007)

\bibitem{Sharifzadeh}
Sharifzadeh, M., Shahabi, C.:
\newblock The spatial skyline queries.
\newblock In: VLDB '06: Proc. of the 32nd International Conference on Very
  Large Data Bases. (2006)  751--762

\bibitem{Kossmann02}
Kossmann, D., Ramsak, F., Rost, S.:
\newblock Shooting stars in the sky: An online algorithm for skyline queries.
\newblock In: VLDB '02: Proc. of the 28th International Conference on Very
  Large Data Bases. (2002)  275--286

\bibitem{Godfrey05}
Godfrey, P., Shipley, R., Gryz, J.:
\newblock Maximal vector computation in large data sets.
\newblock In: VLDB '05: Proc. of the 31st International Conference on Very
  Large Data Bases. (2005)  229--240

\bibitem{ChanED06}
Chan, C.Y., Jagadish, H., Tan, K., Tung, A.K., Zhang, Z.:
\newblock On high dimensional skylines.
\newblock In: EDBT '06: Proc. of the 10th International Conference on Extending
  Database Technology. (2006)

\bibitem{Chan06}
Chan, C.Y., Jagadish, H., Tan, K.L., Tung, A.K., Zhang, Z.:
\newblock Finding k-dominant skylines in high dimensional space.
\newblock In: SIGMOD '06: Proc. of the 2006 ACM SIGMOD International Conference
  on Management of Data. (2006)

\bibitem{Lin07}
Lin, X., Yuan, Y., Zhang, Q., Zhang, Y.:
\newblock Selecting stars: The k most representative skyline operator.
\newblock In: ICDE '07: Proc. of the 23rd International Conference on Data
  Engineering. (2007)  86--95

\bibitem{Rousso}
Roussopoulos, N., Kelley, S., Vincent, F.:
\newblock Nearest neighbor queries.
\newblock In: SIGMOD '95: Proc. of the 1995 ACM SIGMOD international conference
  on Management of data. (1995)  71--79

\bibitem{Berchtold}
Berchtold, S., B\"{o}hm, C., Keim, D.A., Kriegel, H.P.:
\newblock A cost model for nearest neighbor search in high-dimensional data
  space.
\newblock In: PODS '97: Proc. of the 16th ACM SIGACT-SIGMOD-SIGART symposium on
  Principles of database systems. (1997)  78--86

\bibitem{Bohm}
Beyer, K.S., Goldstein, J., Ramakrishnan, R., Shaft, U.:
\newblock When is ''nearest neighbor'' meaningful?
\newblock In: ICDT '99: Proc. of the 7th International Conference on Database
  Theory. (1999)  217--235

\bibitem{Papadias}
Papadias, D., Tao, Y., Mouratidis, K., Hui, C.K.:
\newblock Aggregate nearest neighbor queries in spatial databases.
\newblock Volume~30. (2005)  529--576

\bibitem{Huang}
Huang, X., Jensen, C.S.:
\newblock In-route skyline querying for location-based services.
\newblock In: W2GIS. (2004)

\bibitem{CG}
de~Berg, M., Cheong, O., van Kreveld, M., Overmars, M.:
\newblock Computational Geometry : Algorithms and Applications. Third edn.
\newblock Springer Verlag (2008)

\bibitem{qhull}
:
\newblock Qhull code for convex hull, delaunay triangulation, voronoi diagram,
  and halfspace intersection about a point.
\newblock World Wide Web electronic publication (May 1995)

\bibitem{rstartree}
Beckmann, N., Kriegel, H.P., Schneider, R., Seeger, B.:
\newblock The {R}$^*$-tree: An efficient and robust access method for points
  and rectangles.
\newblock In: SIGMOD '90: Proc. of the 1990 ACM SIGMOD international conference
  on Management of data. (1990)  322--331

\bibitem{Klee80}
Klee, V.:
\newblock On the complexity of d-dimensional {V}oronoi diagrams.
\newblock Archiv der Mathematik \textbf{34} (1980)  75--80

\bibitem{Chazelle91}
Chazelle, B.:
\newblock An optimal convex hull algorithm and new results on cuttings.
\newblock In: Proc. 32nd Annu. IEEE Sympos. Found. Comput. Sci. (1991)  29--38

\bibitem{Clarkson89}
Clarkson, K.L., Shor, P.W.:
\newblock Applications of random sampling in computational geometry.
\newblock {II}. Discrete Comput. Geom. \textbf{4} (1989)  387--421

\bibitem{Seidel91}
Seidel, R.:
\newblock Small-dimensional linear programming and convex hulls made easy.
\newblock Discrete Comput. Geom. \textbf{6} (1991)  423--434

\end{thebibliography}

\end{document}